\documentclass[11pt]{article}

\def\withcolors{1}
\def\withnotes{1}
\usepackage{xspace}
\usepackage{graphicx}
\usepackage{waingarten}
\usepackage{canonne}

\def\FullBox{\hbox{\vrule width 8pt height 8pt depth 0pt}}
\newcommand{\QED}{\;\;\;\FullBox}
\renewenvironment{proof}{\noindent{{\textbf{Proof:}~}}} {\hfill\QED}
\providecommand{\email}[1]{\href{mailto:#1}{\nolinkurl{#1}\xspace}}

\newcommand{\lapl}[1][i]{L_{#1}} %

\newcommand{\SCOND}{\ensuremath{\mathsf{SCOND}}\xspace}
\newcommand{\tc}{$t$-contributing\xspace}
\newcommand{\tpc}{$(t+1)$-contributing\xspace}
\newcommand{\gtc}{$t$-contributing\xspace}
\newcommand{\gtpc}{$(t+1)$-contributing\xspace}
\newcommand{\SubCondUni}{$\textsc{SubCondUni}$\xspace}

\title{Random Restrictions of High-Dimensional Distributions and Uniformity Testing with Subcube Conditioning}

\author {
  Cl\'ement L. Canonne\thanks{IBM Research. \email{ccanonne@cs.columbia.edu}. Part of this work was conducted while a Motwani Postdoctoral Fellow at Stanford University.}
  \and
  Xi Chen\thanks{Columbia University. \email{xichen@cs.columbia.edu}. 
  Supported by NSF IIS-1838154 and NSF CCF-1703925.}
  \and 
  Gautam Kamath\thanks{Cheriton School of Computer Science, University of Waterloo. \email{g@csail.mit.edu}. Supported by a University of Waterloo startup grant. 
  Part of this work was conducted while the author was supported by a Microsoft Research Fellowship, as part of the Simons-Berkeley Research Fellowship program, and while visiting Microsoft Research, Redmond.}
  \and
  Amit Levi\thanks{Cheriton School of Computer Science, University of Waterloo. \email{amit.levi@uwaterloo.ca}. Research supported by the David R. Cheriton Graduate Scholarship.}
  \and
  Erik Waingarten\thanks{Columbia University. \email{eaw@cs.columbia.edu}. Supported by the NSF Graduate Research Fellowship (Grant No. DGE-16-44869)}
}

\begin{document}
\maketitle

\begin{abstract}
We give a nearly-optimal algorithm for testing uniformity of distributions supported on $\{-1,1\}^n$, which makes $\widetilde{O}(\sqrt{\dims}/\dst^2)$ many queries to a subcube conditional sampling oracle (Bhattacharyya and Chakraborty (2018)). The key technical component is a natural notion of random restrictions for distributions on $\bool^\dims$, and a quantitative analysis of how such a restriction affects the mean vector of the distribution. Along the way, we consider the problem of \emph{mean testing} with independent samples and provide a nearly-optimal algorithm.
\end{abstract}

\thispagestyle{empty}
\newpage
\tableofcontents
\thispagestyle{empty}
\newpage
\pagenumbering{arabic}

r%

\section{Introduction}\label{sec:intro}

The focus of this paper is high-dimensional distribution testing. The algorithmic problem is the following: we are granted oracle access (the type of which we will specify shortly) to a probability distribution $\p$ on $\Sigma = \bool^\dims$, and must distinguish with probability at least $2/3$ between the case where $\p$ is the uniform distribution, and that where $\p$ is $\dst$-far from uniform in total variation distance. 
The classical works of distribution testing \cite{GoldreichGR96, GoldreichR00, BatuFRSW00} study the above question in the standard statistical setting, where the oracle provides independent samples from $\p$. In this case, the hallmark results are an algorithm and a matching lower bound, showing that $\Theta(\sqrt{|\Sigma|} / \dst^2)$ independent samples are necessary and sufficient for testing uniformity~\cite{Paninski08, ValiantV14}. When studying distributions supported on high-dimensional domains, unfortunately, this implies that the complexity of sample-optimal algorithms scales exponentially with the dimension, effectively making the problem intractable. To circumvent this issue, recent work has proceeded by either restricting the class of input distributions (e.g., restricting $\p$ to be a product distribution; see~\cref{sec:related}), or by allowing stronger oracle access. We take the latter approach, and consider an oracle access which is particularly well-suited to the high-dimensional structure: the \emph{subcube conditional query model}. 

Subcube conditional query access, first suggested in \cite{CanonneRS15} and studied in \cite{BhattacharyyaC18}, allows algorithms to specify a subcube of the high-dimensional domain and request a sample from the distribution \emph{conditioned} on the sample lying in the subcube specified~---~equivalently, to request samples after fixing some of their variables. The operation is akin to the notion of \emph{restrictions} in the analysis of Boolean functions. Specifically, we identify the distribution by its probability mass function $\p \colon \bool^\dims \to \R_+$. An algorithm may then specify a subcube by a string $\rho \in \{-1,1, *\}^\dims$, where $*$'s denote free variables and non-$*$'s denote the values of restricted variables. Calling the oracle on such a $\rho$ results in a sample from the distribution $\p_{|\rho}$ (now supported on $\{-1,1\}^{\stars(\rho)}$) given by restricting the function $\smash{\p_{|\rho} \colon \{-1,1\}^{\stars(\rho)} \to \R_+}$ and re-normalizing it, so that it represents a distribution $\p_{|\rho}$.\footnote{When conditioning on a subcube with zero support, one may consider models where the oracle returns a uniform sample \cite{ChakrabortyFGM16} or outputs ``error'' \cite{CanonneRS15}. We note that our algorithm will never run into this scenario.}

Our main results are two-fold: (i)~We define a natural notion of random restrictions for high-dimensional distributions and analyze the behavior of the mean vector of a distribution under such random restrictions (see \cref{thm:main-res-intro} in \cref{hehesec}); (ii)~Leveraging this analysis, we obtain a nearly-optimal algorithm for testing uniformity over $\bool^\dims$ with subcube conditioning. As stated below, subcube conditioning allows us to go from $2^{\dims/2} / \eps^2$ sample complexity to $\sqrt{\dims}/\eps^2$: %
\begin{theorem}[Main Result: Uniformity Testing]\label{thm:subcube-alg-intro}
There exists an algorithm which, given subcube conditional query access to a distribution $\p$ supported on $\bool^\dims$ and a distance parameter $\dst \in (0,1)$, makes $\smash{\widetilde{O}(\sqrt{\dims} / \dst^2)}$ queries and can distinguish with probability at least $2/3$ between the case when $\p$ is uniform, and when $\p$ is $\dst$-far from uniform in total variation distance.
\end{theorem}

{\cref{thm:subcube-alg-intro} is tight up to poly-logarithmic factors. Indeed, as observed in~\cite{BhattacharyyaC18}, the sample complexity lower bound of $\Omega(\sqrt{\dims} / \dst^2)$ of~\cite{CanonneDKS17, DaskalakisDK18} for testing uniformity of \emph{product} distributions carries over to subcube conditional sampling.\footnote{{The reason is that, for product distributions, the coordinates are already independent, and therefore conditioning on subcubes does not grant any additional power.}}}
Our result shows that, with subcube conditional queries, testing uniformity over arbitrary distributions is no harder than that over the much more restricted class of {product} ones. 
\medskip

\paragraph{Comparison with \cite{BhattacharyyaC18}.}
{\cref{thm:subcube-alg-intro} improves the
   upper bound of $\tilde O(\dims^2/\dst^2)$ of~\cite{BhattacharyyaC18} 
   for uniformity testing with subcube conditional queries, bringing it to the sublinear regime.
The algorithm of \cite{BhattacharyyaC18} is based on a chain rule that, roughly speaking, 
  bounds the mean of an \emph{individual} coordinate of a distribution 
  after a random restriction.
In contrast, our algorithm applies new machinery (\cref{thm:main-res-intro}) developed to analyze the mean
  \emph{vector} (its $\ell_2$-norm, in particular)  
  after a random restriction.
Along the way, we study the \emph{mean testing} problem, a natural variant of uniformity testing for high-dimensional distributions, and obtain optimal bounds for this question in the standard sampling model.

While our bounds for uniformity testing are quantitatively stronger,
  we note that the algorithm of \cite{BhattacharyyaC18} works for the problem of testing 
  against any known distribution and over any product domains.
Extending our results to these settings is an interesting direction for future work.

\paragraph{Comparison with \cite{chen2020learning}.}
In a simultaneous submission, \cite{chen2020learning} 
  leverages techniques developed in the current paper 
  for analyzing random restrictions and mean testing
  to study the \emph{learning and testing of $k$-junta distributions} (uniformity testing 
  can be viewed as the case when $k=0$).
A set of new algorithmic primitives of independent interest is developed 
  in \cite{chen2020learning} to deal with $k$-junta distributions, and
a substantial component of \cite{chen2020learning} is 
  in (nearly-optimal) lower bounds for learning and testing $k$-junta distributions.
The algorithmic results of \cite{chen2020learning} 
  demonstrate the potential of techniques developed in the current paper 
  for attacking broader learning and testing problems with subcube conditioning.}

\subsection{{Technical Ingredients}}\label{hehesec}

{We start by reviewing the work of \cite{CanonneDKS17} for testing
  uniformity over product distributions.}\vspace{-0.1cm}

\paragraph{Product Distributions and Mean Distance.} 

The simplest class of distributions on $\bool^\dims$ is arguably the class of \emph{product distributions}, where all coordinates are independent. This setting was studied in \cite{CanonneDKS17}, and is particularly nice to analyze due to the relation between the total variation distance between distributions and the $\lp[2]$ distance  between their mean vectors. Specifically, let $\p$ be a product distribution supported on $\bool^\dims$, and $\mu(\p) \in [-1,1]^\dims$ be its \emph{mean vector}, 
\[ \mu(\p) = \Ex_{\bx \sim p}[\bx] \in [-1,1]^\dims. \]
It is not hard to show that if $\p$ is $\dst$-far from uniform in total variation distance, then $\normtwo{\mu(\p)} \gsim \dst$. 
Hence, for product distributions, 
large total variation distance to uniformity implies large mean vector in $\lp[2]$ norm. Given this fact, Canonne, Diakonikolas, Kane, and Stewart~\cite{CanonneDKS17} design an algorithm based on estimating the norm of the mean vector, and show that $O(\sqrt{\dims} / \dst^2)$ many samples from product distributions suffice to test uniformity.\footnote{In fact, they consider the more general problem of identity testing of product distributions, where they obtain analogous results.} 

However, the relationship observed between distance to uniformity in total variation and \emph{mean distance} (i.e., the $\lp[2]$ norm of the mean vector) for product distributions is not true in general. A simple example is the uniform distribution supported on just two vectors $\{x, -x\}$, which is very far from uniform yet has mean vector $0$. Towards relating these two notions for general distributions, we define our notion of random restriction. 

\paragraph{Random Restrictions.}  For any $\sigma \in [0,1]$, we write $\calS_{\sigma}$ for the distribution supported on subsets of $[\dims]$ given by letting $\bS \sim \calS_{\sigma}$ include each index $i \in \bS$ independently with probability $\sigma$. Given any distribution $\p$ supported on $\bool^\dims$, let $\calD_{\sigma}(\p)$ be the distribution, supported on $\{-1,1, *\}^\dims$, of \emph{random restrictions} of $\p$. In order to sample a random restriction $\brho \sim \calD_{\sigma}(\p)$, we sample a set $\bS \sim \calS_{\sigma}$ and a sample $\bx \sim \p$; then, we let $\brho_i$ be set according to:
\begin{align} \label{eq:rr-intro}
\brho_i =
\begin{cases}
  * & \text{if } i \in \bS \\
  \bx_i & \text{if } i \notin \bS
\end{cases}.
\end{align}
For any $\rho \in \{-1,1,*\}^\dims$, we denote by $\p_{|\rho}$ the distribution on $\bool^{\stars(\rho)}$ given by $\bx_{\stars(\rho)}$ where $\bx$ is drawn from $\p$ \emph{conditioned} on every $i \notin \stars(\rho)$ being set to $\rho_i$. This defines the restriction~of a distribution $\p$. Another operation on distributions we will consider is that of \emph{projection}; for any set $S \subset [\dims]$, we write $\ol{S} = [n] \setminus S$, and the distribution $\p_{\ol{S}}$ supported on $\{-1,1\}^{\ol{S}}$ is given by letting $\by \sim \p_{\ol{S}}$ be $\by = \bx_{\ol{S}}$ for $\bx \sim p$. 

Let $\dtv(p,\eUniform)$ denote the total variation distance between $p$ and the uniform distribution
  of the same dimension.
At a high level, our main technical result shows that \emph{the mean distance of random restrictions is implied by total variation distance of random projections}.

\begin{theorem}[Informal version; see \cref{thm:main-restriction}]\label{thm:main-res-intro}
Let $\p$ be any distribution over $\bool^\dims$. Then, 
\begin{align}
\bE{\brho \sim \calD_{\sigma}(\p)}{ \normtwo{ \mu(\p_{|\brho}) } } &\geq \sigma\cdot \Ex_{\bS \sim \calS_{\sigma}}\big[\dtv(\p_{\ol{\bS}}, \eUniform)\big] . \label{eq:projection-to-restriction}
\end{align}
\end{theorem}
Although the above differs from \cref{thm:main-restriction} in certain respects 
  (the inequality in \cref{thm:main-restriction} incurs additional poly-logarithmic factors as well as a small additive error), (\ref{eq:projection-to-restriction}) captures the key relationship between the total variation distance and the mean distance that we leverage, and will provide intuition for the introduction.

\paragraph{Mean Testing.}
The above discussion naturally leads to the following problem, which we refer to as \emph{mean testing}. Given sample access to a distribution $\p$ supported on $\bool^\dims$, we seek to distinguish with probability at least $2/3$ between the case where $\p$ is uniform and that where $\p$ has a large mean vector, i.e., $\normtwo{\mu(\p)} \geq \dst\sqrt{\dims}$. When $\p$ is assumed to be a product distribution, \cite{CanonneDKS17, DaskalakisDK18} showed that for $\dst \leq 1/\sqrt{\dims}$ the sample complexity of the problem is $\Theta(1/(\dst^2\sqrt{\dims}))$.\footnote{This is implicit in \cite[Theorem~4.1]{CanonneDKS17}. As stated, their result is suited for distinguishing between the mean vector having norm $0$ or at least $\dst'$, where $\dst' \in (0,1]$: we re-parameterize with $\dst' = \dst \sqrt{\dims}$.} The idea is that the empirical mean of a product distribution will have norm concentrated around its true value, and thus it suffices to look at this empirical estimate. In our case, however, additional care is needed, as there can be arbitrary correlations between coordinates of a sample and this concentration does not hold in general. Nevertheless, we present an algorithm for mean testing which is optimal up to a triply-logarithmic small loss in the sample complexity.
\begin{theorem}[Mean Testing]\label{thm:mean-testing-intro}
There exists an algorithm which given 
\[ \bigO{ \max\left\{\frac{1}{\dst^2 \sqrt{\dims}}, \frac{1}{\dst} \right\} } \]
 i.i.d.\ samples from an arbitrary distribution $\p$ on $\bool^\dims$ and a parameter $\dst \in (0, 1]$ can distinguish with probability at least $2/3$ between (i)~$\p$ is the uniform distribution, and (ii)~$\normtwo{\mu(\p)} \geq \dst\sqrt{\dims}$.
\end{theorem}
Moreover, as detailed in~\cref{ssec:testing:mean:gaussian}, the above immediately implies a similar sample complexity for \emph{Gaussian} mean testing, where one is given i.i.d.\ samples from a multivariate normal distribution $\p=\gaussian{\mu}{\Sigma}$ and must distinguish between $\p=\gaussian{0}{I}$ and $\normtwo{\mu}\geq \dst\sqrt{\dims}$.

\paragraph{Uniformity Testing with Subcube Conditioning.}
In view of the above discussion, we aim to use random restrictions and \cref{thm:mean-testing-intro} to test uniformity with subcube conditional queries. The final technical ingredient is the following inequality, very similar to the ``chain rule'' of Bhattacharyya and Chakraborty~\cite{BhattacharyyaC18} suited for uniformity testing on $\bool^\dims$.
\begin{lemma}\label{lemlem1}Let $\p$ be a distribution supported on $\bits^\dims$. Then, for any $\sigma\in [0,1]$, 
	\[
	    \dtv(\p,\Uniform)\le\Ex_{\bS \sim \calS_\sigma}\big[\dtv(\p_{\overline{\bS}},\Uniform)  \big]+\Ex_{ \brho\sim \calD_\sigma(\p)} \big[\dtv(\p_{|\brho },\Uniform)\big].
	\]
\end{lemma}
This lemma naturally leads to a recursive approach for uniformity testing. Given a distribution $\p$ which is $\dst$-far from uniform, either the total variation of random projections is large, or the total variation of random restrictions is large. In the former case, we apply (\ref{eq:projection-to-restriction}), which allows us to reduce the problem to that of mean testing, and invoke \cref{thm:mean-testing-intro}. In the latter case, we take a random restriction and recurse (but on far fewer variables).

\subsection{{Proof Overview}}

We now formally state and explain the intuition behind our main theorem, relating the distance to uniformity of random projections to the mean distance after random restrictions.

\begin{theorem}\label{thm:main-restriction}
Let $\p$ be any distribution over $\bool^\dims$ and $\sigma\in [0,1]$. Then, 
\[ \bE{\brho \sim \calD_{\sigma}(\p)}{ \normtwo{ \mu(\p_{|\brho}) } } \geq \frac{\sigma}{\poly(\log n)} \cdot \widetilde{\Omega}\left(\Ex_{\bS \sim \calS_{\sigma}}[\dtv(\p_{\ol{\bS}}, \eUniform)] - {2}\hspace{0.02cm}e^{-\min(\sigma,1-\sigma) n/10} \right). \]
\end{theorem}

We encourage the reader to think of applying \cref{thm:main-restriction} to a distribution $\p$ which is $\dst$-far from uniform, and to think of the case when the parameter $\sigma$ is a small constant (or inverse of a poly-logarithmic factor), and $\Ex_{\bS \sim \calS_{\sigma}}[\dtv(\p_{\ol{\bS}}, \Uniform)] = \Theta(\dst)$. Specifically, consider a parameter setting where $e^{-\min(\sigma,1-\sigma) n/10} = o(\dst)$ so that the right-hand side of the expression in \cref{thm:main-restriction} becomes $\smash{\dst \sigma /\poly(\log n, \log(1/\dst))}$.

The proof of \cref{thm:main-restriction} proceeds by proving a lemma (\cref{lem:main-tech}) which captures the behavior of random restrictions with $t$ stars versus those with $t+1$ stars. 
We prove a robust version of Pisier's inequality \cite{Pisier86}, 
  an inequality which was first studied in the geometry of Banach spaces.
  When the projections $\p_{\ol{S}}$ are far from uniform in total variation, the robust version of Pisier's inequality will help us lower bound certain quantities of a collection of 
    directed graphs defined using $p_{\ol{S}}$ on the subcube $\smash{\{-1,1\}^{\ol{S}}}$. The coordinate values of the mean vector of $\p_{|\brho}$ will depend on whether random vertices 
    and directions induced from $\brho$ are directed edges of one of the graphs in the collection.
    Hence, the lower bound on expected norm of the mean vector will follow from studying 
   structures of directed graphs in the collection. 

We first introduce the robust version of Pisier's inequality. As a warm-up and in order to show the usefulness of this inequality, we give an algorithm making $\widetilde{O}(n/\dst^2)$ queries which we refer to as an \emph{edge tester}. The reason for this name is that the algorithm samples a random restriction~$\brho$  with exactly one star, so that the distribution $\p_{|\brho}$ is supported on an edge of the hypercube. Then, the algorithm tests whether $\p_{|\brho}$ is uniform. Lastly, we provide a proof sketch for a weaker version of \cref{thm:main-restriction} which lower bounds the expectation of $\| \mu(\p_{|\brho}) \|_2^2$, rather than $\| \mu(\p_{|\brho}) \|_2$. The weaker inequality is insufficient for our purposes, but the proof is conceptually much simpler (since $\|\cdot\|_2^2$~is additive over coordinates and thus, we may use linearity of expectation). To prove \cref{thm:main-restriction}, attempting to lower bound $\normtwo{ \mu(\p_{|\brho}) }$ reveals further challenges which necessitate additional care.

\subsubsection{A Robust Pisier's Inequality}

We let $\calH$ denote the set of undirected edges of the hypercube $\{-1,1\}^n$. For a function $f\colon \{-1,1\}^n \to \R$ and $i \in [n]$, we write $\smash{\lapl[i] f(x) \eqdef (f(x) - f(x^{(i)}))/2}$. The following inequality is known as Pisier's inequality \cite{Pisier86}. We state it below in a way most closely to how it will be applied in this paper; the inequality holds in much larger generality for functions $f \colon \{-1,1\}^n \to X$ over general Banach spaces $X$ (see, in particular, \cite{NaorS02}).
\begin{theorem}[Pisier's inequality]
  \label{thm:vanilla-pisier}
Let $f\colon \{-1,1\}^n \to \R$ be a function with $\Ex_{\bx}[f(\bx)] = 0$. Then,
\[ \Ex_{\bx \sim \{-1,1\}^n}\big[|f(\bx)|\big] \lsim \log n \cdot \Ex_{\bx, \by \sim \{-1,1\}^n}\left[\hspace{0.03cm}\left| \sum_{i=1}^n \by_i \bx_i \lapl[i]f(\bx) \right| \hspace{0.03cm}\right]. \]
\end{theorem} 

In this paper, we will need a \emph{robust} version of Pisier's inequality in order to derive \cref{thm:main-restriction}. The notion of robustness is equivalent to the notion considered in \cite{KhotMS18}, who proved robust (and directed) versions of Talagrand's inequality for Boolean functions, and part of our proof utilizes the robustness of the inequality in a similar way. Specifically, we consider an arbitrary orientation of the edges of the hypercube and sum the values of $\lapl[i]f(\bx)$ only when the edge $\{\bx, \bx^{(i)}\}$ is oriented from $\bx$ to $\bx^{(i)}$. 

\begin{theorem}[Robust version of Pisier's inequality]\label{thm:pisier-intro}
Let $f\colon \{-1,1\}^n \to \R$ be a function satisfying $\Ex_{\bx}[|f(\bx)|] = 0$. Let $G = (\{-1,1\}^n, E)$ be any orientation of the hypercube. Then, 
\[ \Ex_{\bx \sim \{-1,1\}^n}\big[|f(\bx)|\big] \lsim \log n \cdot \Ex_{\bx, \by \sim \{-1,1\}^n}\left[\hspace{0.04cm} \left| \sum_{\substack{i \in[n] \\ (\bx, \bx^{(i)}) \in E}} \by_i \bx_i \lapl[i]f(\bx)  \right|\hspace{0.04cm}\right].\]
\end{theorem}
The proof  follows the template of~\cite[Theorem 2]{NaorS02}, and checks that the necessary changes, even after considering directed edges, still give the desired inequality. 

\begin{remark}We note that Talagrand \cite{Talagrand93} prove that the $\log n$ factor in Pisier's inequality (Theorem~\ref{thm:vanilla-pisier}) is unnecessary for real-valued functions, but that it is for general Banach spaces. While one may follow Talagrand's proof to remove the $\log n$ factor in Theorem~\ref{thm:pisier-intro}, it is not immediately clear whether this approach can handle arbitrary orientations. 
\end{remark}

\subsubsection{Warmup: A Linear Query Algorithm}\label{sec:warmup} To see why~\cref{thm:pisier-intro} is helpful, we give a simple (albeit suboptimal) $\widetilde{O}({\dims/\dst^2})$-query  algorithm for testing uniformity. Suppose $\p$ is a distribution supported on $\{-1,1\}^\dims$ which is $\dst$-far from uniform in total variation distance. We apply~\cref{thm:pisier-intro} to the function $\smash{f(x)\eqdef 2^\dims \cdot \p(x)-1}$, 
where the directed graph $G = ( \{-1,1\}^n , E)$ is given by letting
\[
E = \left\{ (x, x^{(i)}) : \{ x, x^{(i)}\} \text{ is an edge of $\{-1,1\}^n$, and $p(x) \geq p(x^{(i)})$} \right\}.
\]
\cref{thm:pisier-intro} applied on $f$ implies that\footnote{Our proof of \cref{thm:main-restriction} will rely on this strengthening of Pisier's inequality for multiple reasons, as our approach crucially uses the ability to pick any orientation of the edges. 
Even for this simple application one can observe that, without the strengthening, the right-hand side of (\ref{eq:apply-pisier}) would be replaced by 
  the same expression but without the $^+$.
The discussion below (\ref{eq:expected-marginal-intro}) and the derivation of (\ref{eqeq2}) 
  explains why having the $^+$ in the 
  expression is crucial.
}
\begin{align}
    \Ex_{\bx\sim\bool^\dims}\big[|f(\bx)|\big] &\lsim \log \dims \cdot \Ex_{\bx,\by\sim\bool^\dims}\left[\hspace{0.03cm}\abs{\hspace{0.04cm}\sum_{i=1}^\dims \by_i \bx_i\big(f(\bx)-f(\bx^{(i)})\big)^+}\hspace{0.03cm}\right]. \label{eq:apply-pisier}
\end{align}
Given (\ref{eq:apply-pisier}), we start by applying the fact that the left-hand side of the inequality is at least $2\dst$.
From there, the following three (in)equalities are obtained via Khintchine's inequality, importance sampling, and Jensen's inequality, respectively. (We also use the convention that $0/0=0$.)
\begin{align}
 \frac{\dst}{\log\dims}
 \lesssim  \Ex_{\bx\sim\bool^\dims}\left[\sqrt{\sum_{i=1}^\dims \left(\big(f(\bx)-f(\bx^{(i)})\big)^+\right)^2}\right]
 &= \Ex_{\bx\sim\p}\left[\sqrt{\sum_{i=1}^\dims \left(\frac{(f(\bx) - f(\bx^{(i)}))^+}{1+f(\bx)}\right)^2}\right] \nonumber \\
 & \leq \left(\Ex_{\bx\sim\p}\left[\sum_{i=1}^\dims \left(\frac{(p(\bx) - p(\bx^{(i)}))^+}{p(\bx)}\right)^2\right]\right)^{1/2}. \label{eq:expected-marginal-intro}
\end{align}
Notice that for every $x$ and $i$, either $(p(x) - p(x^{(i)}))^+=0$ or $p(x)>p(x^{(i)})\ge 0$.
First, this makes sure that $a/0$ with $a>0$ never occurs in (\ref{eq:expected-marginal-intro}).
Additionally,
\begin{equation}\label{eqeq2}
 0 \leq \frac{1}{2} \cdot \dfrac{(p(x) - p(x^{(i)}))^+}{p(x)} \leq \dfrac{|p(x) - p(x^{(i)})|}{p(x) + p(x^{(i)})}, 
\end{equation}
and the quantity on the right-hand side is exactly the bias on the distribution of the $i$th bit of a draw of $\p$ conditioning on all $i' \neq i$ bits set according to $x$. 
 In particular, 
 a standard averaging/bucketing argument shows that there exists $\beta \geq \dst^2 / (n\log^2 n)$ such that
\begin{equation}
    \Prx_{\substack{\bx\sim\p\\\bi\sim[\dims]}}\left[ \left(\frac{|\p(\bx) - \p(\bx^{(\bi)})|}{\p(\bx) + \p(\bx^{(\bi)})}\right)^2 \gsim \frac{\dst^2}{\beta \cdot n \log^2 n \log(n/\dst)}  \right] \geq \beta. \label{eq:good-event}
\end{equation}

The above lower bound suggests the following ``edge tester:'' for all $h \in \{0, \dots, O(\log(n/\dst))\}$ where $2^{-h} \gsim \dst^2 / (n\log^2 n)$, independently sample ${O}(2^{h}\log(n/\dst))$ pairs $(\bx,\bi)$ ($\bx$ from $\p$, and $\bi \sim [\dims]$). 
For each pair $(\bx, \bi)$, we consider the distribution supported on $\{-1,1\}$ given by sampling $\by \sim \p$ \emph{conditioned} on every $i' \neq \bi$ having $\by_{i'} = \bx_{i'}$, and use $\smash{\widetilde{O}( 2^{-h} n / \dst^2)}$ queries to estimate the bias of the conditional distribution up to error $(\dst \cdot (2^h/n)^{1/2})/\polylog(n/\eps)$ with high probability. If $\p$ was uniform, every conditional distribution considered will be uniform; however, if $\p$ is $\dst$-far from uniform in total variation distance, (\ref{eq:good-event}) implies that some setting of $h$ will reveal a large bias with high probability. The query complexity of this algorithm is
\[ \sum_{h=0}^{O(\log(n/\dst))} O\left(2^h \cdot \log\left(\frac{n}{\dst}\right)\right) \cdot \widetilde{O}\left({ \frac{2^{-h} n}{\dst^2} }\right) = \widetilde{O}\left({\frac{\dims}{\dst^2}}\right). \]

\subsubsection{A Weaker Version of \cref{thm:main-restriction}}

In order to highlight some of the conceptual ideas involved in proving \cref{thm:main-restriction}, we sketch how one may prove the following weaker inequality, which lower bounds the expected squared $\lp[2]$ norm instead of the expected $\lp[2]$ norm:
\begin{align}
\Ex_{\brho \sim \calD(t+1, \p)}\left[\normtwo{ \mu(\p_{|\brho}) }^2\right] \gsim \frac{1}{\log^2 n} \cdot \frac{t+1}{n-t} \cdot \Ex_{\substack{\bT \subset [n] \\ |\bT| = t}}\left[ \dtv(\p_{\ol{\bT}}, \Uniform)^2\right]   \label{eq:ell-2-2}
\end{align}
where $\calD(t+1, \p)$ denotes the distribution over random restrictions where we enforce $\stars(\brho) = t+1$ (to compare this to \cref{thm:main-restriction}, one should think of $\sigma$ as $t/n$). Specifically, we sample a random set $\bS \subset [n]$ of size $t+1$ and $\bx \sim \p$, then, we set $\brho$ as in (\ref{eq:rr-intro}). Consider imposing a random order $\bpi \colon [t+1] \to \bS$ and apply linearity of expectation to re-write the left-hand side of (\ref{eq:ell-2-2}) as 
\begin{align} 
\sum_{i=1}^{t+1} \Ex_{\brho, \bpi}\left[(\mu(\p_{|\brho})_{\bpi(i)})^2\right] = (t+1) \Ex_{\brho, \bpi}[\mu(\p_{|\brho})_{\bpi(t+1)}^2]. \label{eq:ell-2-2-2}
\end{align}
Toward lower bounding (\ref{eq:ell-2-2-2}), consider the following way of sampling  $\brho$ and $\bpi$. We first sample $t$ random indices $\bpi(1), \dots, \bpi(t)$ uniformly from $[n]$ without replacement and let $\bT = \{ \bpi(1), \dots \bpi(t)\}$. Then, we sample $\bx \sim \p$. Finally, we sample $\bpi(t+1)$ uniformly from the set $\ol{\bT} = [n] \setminus \bT$. We may write $\bS = \bT \cup \{\bpi(t+1)\}$ and similarly have $\brho$ be set according to (\ref{eq:rr-intro}). Consider the function $\bell \colon \{-1,1\}^{\ol{\bT}} \to [-1, \infty)$ which is given by 
\begin{align} 
\bell(y) &= 2^{|\ol{\bT}|} \Prx_{\by \sim \p_{\ol{\bT}}}[\by = y] - 1. \label{eq:2-2-3}
\end{align}
We notice that $\bell$ has mean zero (since $\p_{\ol{\bT}}$ is a probability distribution), similarly to~\cref{sec:warmup}, we orient the edges of the hypercube $\smash{\{-1,1\}^{\ol{\bT}}}$ from the endpoint with higher value of $\bell$ to lower value of $\bell$. Applying \cref{thm:pisier-intro} to $f = \bell$ and this orientation of the hypercube edges, the left-hand side is exactly $2\hspace{0.02cm}\dtv(\p_{\ol{\bT}}, \Uniform)$. Further, we write the random variable $\bz = \bx_{\ol{\bT}}$, which is distributed exactly as $\bz \sim \p_{\ol{\bT}}$, and the random variable $\bj = \bpi(t+1)$, which is distributed uniformly from $\ol{\bT}$. We have, similarly to (\ref{eq:expected-marginal-intro}) but with $f = \bell$ rather than $\p$,
\begin{align}
\frac{\dtv(\p_{\ol{\bT}}, \Uniform)}{\log n} &\lsim \left(|\ol{\bT}| \cdot \Ex_{\substack{\bz \sim \p_{\ol{\bT}} \\ \bj \sim \ol{\bT}}}\left[\left(\dfrac{(\bell(\bz) - \bell(\bz^{(\bj)}))^+}{1+\bell(\bz)} \right)^2 \right]\right)^{1/2}.\label{eq:expectation-2-2}
\end{align}
One crucial observation is that the inner value of the expectation in the right-hand side of (\ref{eq:expectation-2-2}) is, up to a factor of at most $4$, $\mu(\p_{|\brho})_{\bpi(t+1)}^2$, where $\brho$ uses $\bT \cup \{\bj\}$ as stars and non-star values are set to $\bz_{-\bj}$. Plugging this back into (\ref{eq:ell-2-2-2}) gives (\ref{eq:ell-2-2}). 
 
\subsection{Related Work}\label{sec:related}
The seminal works of Goldreich, Goldwasser, and Ron~\cite{GoldreichGR96}, Goldreich and Ron~\cite{GoldreichR00}, and Batu, Fortnow, Rubinfeld, Smith, and White~\cite{BatuFRSW00} initiated the study of distribution testing, viewing probability distributions as a natural application for property testing (see~\cite{Goldreich17} for coverage of this much broader field).
Since these works, distribution testing has enjoyed a wealth of study, resulting in a thorough understanding of the complexity of testing many distribution properties (see, e.g.,~\cite{BatuFFKRW01, BatuKR04, Paninski08, AcharyaDJOP11, BhattacharyyaFRV11, Valiant11, IndykLR12, DaskalakisDSVV13, ChanDVV14, ValiantV17a, Waggoner15, AcharyaDK15, BhattacharyaV15, DiakonikolasKN15a, DiakonikolasK16, Canonne16, BlaisCG17, BatuC17, DaskalakisKW18, DiakonikolasGPP18}, and~\cite{Rubinfeld12,Canonne15a, BalakrishnanW18, Kamath18} for recent surveys).
As a result, sample-optimal algorithms are known for a number of core problems. 

However, the known sample complexity lower bounds, while sublinear, typically still involve a polynomial dependence on the domain size, suggesting that new models are needed for studying distribution testing on high-dimensional domains. One approach is to assume additional \emph{structure} from the input distributions. Settings were studied where the distribution is known to be monotone~\cite{RubinfeldS09}, a low-degree Bayesian Network~\cite{CanonneDKS17, DaskalakisP17, AcharyaBDK18}, a Markov Random Field~\cite{DaskalakisDK18, GheissariLP18, BezakovaBCSV19}, or having some ``flat'' histogram structure~\cite{DiakonikolasKP19}. 

The other approach is to allow stronger oracle access. The subcube conditional sampling model, which is the focus of this work, is a variant of the general conditional sampling model particularly apt for the study of high-dimensional distributions. Bhattacharyya and Chakraborty~\cite{BhattacharyyaC18}, who initiated the systematic study of this variant, showed that for many problems of interest such as uniformity, identity, and closeness testing, subcube conditional queries enabled one to avoid the curse of dimensionality, and established sample complexity upper bounds \emph{polynomial} in the dimension (albeit superlinear). 
The conditional sampling model itself, which was introduced simultaneously by Chakraborty, Fischer, Goldhirsh, and Matsliah~\cite{ChakrabortyFGM13,ChakrabortyFGM16}, and Canonne, Ron, and Servedio~\cite{CanonneRS14, CanonneRS15}, allows more general queries: namely, the algorithm may specify an arbitrary subset of the domain and request a sample conditioned on it lying in the subset. In many cases, the conditional sampling model circumvents sample-complexity lower bounds. Since its introduction, there has been significant study into the complexity of testing a number of properties of distributions under conditional samples, in both adaptive and nonadaptive settings~\cite{Canonne15b, FalahatgarJOPS15, AcharyaCK15b, FischerLV17, SardharwallaSJ17, BlaisCG17, BhattacharyyaC18, KamathT19}.
Beyond distribution testing, this model of conditional sampling has found applications in group testing~\cite{AcharyaCK15a}, sublinear algorithms~\cite{GouleakisTZ17}, and crowdsourcing~\cite{GouleakisTZ18}.
Other ways 
to augment the power of distribution testing algorithms 
include letting the algorithm query the probability density function (PDF) or cumulative distribution function (CDF) of the distribution~\cite{BatuDKR05, GuhaMV06, RubinfeldS09, CanonneR14}, or giving it probability-revealing samples~\cite{OnakS18}.

\subsection{Notation and Prelimaries}

We use boldface symbols to represent random variables, and non-boldface symbols for fixed values (potentially realizations of these random variables) --- see, e.g., $\brho$ versus $\rho$. Given a set $S\subseteq [n]$,~we let $\calU_S$ denote the uniform distribution over $\bits^S$. Usually, as the support of $\calU_{S}$ will be clear from the context, we will drop the subscript and simply write $\calU$.
We write $f(n) \lesssim g(n)$ if, for some $c > 0$, $f(n) \leq c \cdot g(n)$ for all $n \geq 0$ (the $\gtrsim$ symbol is defined similarly). 
$f(n) \asymp g(n)$ if $f(n) \lesssim g(n)$ and $f(n) \gtrsim g(n)$. We use the notation $\tilde{O}(f(n))$ as $O(f(n) \polylog(f(n)))$, and $\tilde{\Omega}(f(n))$ to denote $\Omega(f(n) / (1+ |\polylog(f(n))|))$.
The notation $[k]$ denotes the set of integers $\{1, \dots, k\}$.

The Frobenius norm of a matrix $M \in \mathbb{R}^{d_1 \times d_2}$ is $$\norm{M}_F = \left(\sum_{i\in [d_1]} \sum_{j\in [d_2]} M_{ij}^2\right)^{1/2}.$$ For a string $x \in \bool^\dims$, we use \smash{$x^{(i)}$} to denote the string that is identical to $x$ but with coordinate $i$ flipped, i.e., \smash{$x^{(i)}_j = x_j$} for all $j\neq i$, and \smash{$x^{(i)}_i = -x_i$}. 

We formally define the subcube conditional query access, which was suggested in Canonne, Ron, Servedio~\cite{CanonneRS15} as an instance of the general conditional sampling oracle \cite{CanonneRS15, ChakrabortyFGM16}, and first explicitly studied in Bhattacharyya and Chakraborty~\cite{BhattacharyyaC18}. 
\begin{definition}
  A \emph{subcube conditional sampling} (\SCOND) oracle for a distribution $p$ supported on $\bool^\dims$ is an oracle which accepts a query subcube $\rho \in \{-1,1, *\}^n$, and outputs a sample from the distribution $\bx \sim \p$ conditioned on every $i \notin \stars(\rho)$ having $\bx_i = \rho_i$. We use the convention that if the algorithm considers a restriction with zero support, the oracle outputs a uniform sample.
\end{definition}

\section{The Algorithm}

We prove our theorem for uniformity testing with subcube conditioning, restated below:
\begin{theorem} \label{theo:subconduni}
There exists an algorithm $\SubCondUni$ which, given 
   $n\ge 1$, a subcube oracle to a distribution $\p$ over $\bits^\dims$ and a distance parameter $\dst$
    with $ \dst \in (0,1)$, has the following guarantees. The algorithm makes 
$ \widetilde{O} ( {\sqrt{\dims}}/{\dst^2} ) $ many calls to the oracle and satisfies the following two conditions:\vspace{0.15cm}
	\begin{enumerate}[(i)]
		\item If $\p$ is uniform, then the algorithm returns \accept with probability at least 
		  $2/3$.\vspace{0.05cm}
    \item If $\dtv(\p,\Uniform)\ge\dst$, then the algorithm returns \reject with probability at least 
      $2/3$. 
	\end{enumerate}
\end{theorem}
The rest of this section is devoted to the proof of~\cref{theo:subconduni}.
For our convenience of working with~$\log(1/\dst)$ in the proof, we assume below that $\dst\le 1/2$ in the input of $\SubCondUni$.
This~way $\log(c/\dst)$ can be treated as $O(\log (1/\dst))$ whenever $c\ge 1$ is a fixed constant.
As discussed earlier in Section \ref{sec:intro},
  $\SubCondUni$ (see \cref{algo:subconduni}) is based on \cref{thm:main-restriction} and Lemma~\ref{lemlem1}, which we now prove.

\begin{proofof}{\cref{lemlem1}}
    Fix any subset $S\subseteq [\dims]$ of size $t$. Given $u\in \bits^{\ol S}$, we let $\p_{|\rho(S, u)}$ denote the distribution supported on $\smash{\bits^{S}}$ given by drawing $\bx\sim \p$ conditioned on $\bx_{\ol S}=u$.
    By unrolling the definition of total variation distance, we have  
    \begin{align*}
        2\hspace{0.03cm}\dtv(\p,\Uniform) &= \sum_{x\in\bool^\dims} \abs{\p(x) - 1/2^\dims}\\
        &= \sum_{u\in\bool^{\ol{S}}}\sum_{v\in\bool^{S}} \abs{\Prx_{\bx\sim \p}\big[\bx_{\ol S}=u\hspace{0.03cm}\wedge\hspace{0.03cm} \bx_S=v\big] - 1/2^\dims}\\
        &= \sum_{u\in\bool^{\ol{S}}}\sum_{v\in\bool^{S}} \abs{\Prx_{\bu\sim \p_{\ol S}}\big[\bu= u\big]\cdot\Prx_{\bx\sim \p}\big[\bx_S=v \mid\bx_{\ol S}=u \big] - 1/2^\dims} \\
        &\leq \sum_{u\in\bool^{\ol{S}}}\sum_{v\in\bool^{S}} \Paren{ \abs{\p_{\ol S}(u)\cdot \Prx_{\bx\sim \p}\big[\bx_S=v \mid\bx_{\ol S}=u \big] - \frac{\p_{\ol{S}}(u)}{2^{t}}} +  \abs{\frac{\p_{\ol{S}}(u)}{2^{t}} - \frac{1/2^{\dims-t}}{2^t}} }\\
        &= \sum_{u\in\bool^{\ol{S}}}\p_{\ol{S}}(u)\cdot 2\hspace{0.03cm}\dtv\big(\p_{|\rho(S,u)},\Uniform \big) + \sum_{v\in\bool^{S}} \frac{1}{2^{t}} \cdot 2\hspace{0.03cm}\dtv\big(\p_{\ol{S}},\Uniform \big) \\
        &= \sum_{u\in\bool^{\ol{S}}}\p_{\ol{S}}(u)\cdot 2\hspace{0.03cm}\dtv\big(\p_{|\rho(S,u)},\Uniform \big) + 2\hspace{0.03cm}\dtv\big(\p_{\ol{S}},\Uniform \big).
    \end{align*}
 Taking the expectation of the inequality over the choice of $\bS\sim \calS_\sigma$ yields the lemma.
\end{proofof}\medskip

Let $p$ be a distribution over $\{-1,1\}^n$ with $\dtv(\p,\Uniform)\ge \dst$.
Let $$\sigma\eqdef{\sigma(\dst)}=\frac{1}{C_0 \cdot \log^4 (16/\dst)}$$
  where $C_0>0$ is an absolute constant. (The value of $C_0$ is only used at the end of this section, {where setting $C_0=10^{11}$ is good enough.})
Let us further assume that   $n$ and $\dst$ together satisfy
\begin{equation}\label{smallndef}
e^{-\sigma n/10}\le \dst/{8}.
\end{equation}
Violation of (\ref{smallndef}) 
  implies that  
\begin{equation}\label{hah1}
n=O\left(\frac{1}{\sigma}\cdot \log \left(\frac{{1}}{\dst}\right)\right)=O\left(\log^5\left(\frac{{1}}{\dst}\right)\right)
\end{equation}
and we handle this case by applying the linear-query tester described in~\cref{sec:warmup}
  with query complexity $\smash{\widetilde{O}(n/\dst^2)=\widetilde{O}(1/\dst^2)}$ using (\ref{hah1}).

For the general case with (\ref{smallndef}) satisfied,
  consider a distribution $p$ with $\dtv(\p,\Uniform)\ge \dst$. 
\cref{lemlem1} implies that either $\Ex_{\bS \sim \calS_\sigma}\left[\dtv(\p_{\overline{\bS}},\Uniform)  \right]\geq \dst/2$ or
\begin{equation}\label{eq222} 
\Ex_{ \brho\sim \calD_\sigma(\p)} \left[\dtv(\p_{|\brho },\Uniform)\right]\geq \dst/2.
\end{equation}
Assuming the former and using (\ref{smallndef}), we have from \cref{thm:main-restriction} 
   that 
 (using $\sigma=1/\polylog (1/\dst)$)
\begin{equation}\label{eqsimplified2}
\bE{\brho \sim \calD_{\sigma}(\p)}{ \normtwo{ \mu(\p_{|\brho}) }\big/\sqrt{n} } \geq  \widetilde{\Omega}\left(\frac{\dst}{\sqrt{n}} \right). 
\end{equation}
This naturally lends itself to a recursive approach: for the general case when $n$ and $\dst$ satisfy (\ref{smallndef}), we use $\textsc{TestMean}$ from \cref{thm:mean-testing-intro} to test (\ref{eqsimplified2}), and recursive calls to $\textsc{SubCondUni}$ to test (\ref{eq222}). 

The description of $\textsc{SubCondUni}$ is given in \cref{algo:subconduni}. For convenience, we let 
\begin{equation}\label{uuuu4}
L\eqdef L(n,\dst)=\widetilde{O}(\sqrt{n}/\dst)
\end{equation}
  such that the right hand side of (\ref{eqsimplified2}) 
  can be replaced by $1/L$.

  We are now ready to prove \cref{theo:subconduni}.\medskip

\begin{algorithm}[t]
\algblock[Name]{StartBaseCase}{EndBaseCase}
\algblock[Name]{StartMainCase}{EndMainCase}
	\begin{algorithmic}[1]
		\Require Dimension $\dims$, oracle access to distribution $\p$ over $\bool^\dims$, and parameter $\dst\in (0,1/2]$ 
		\StartBaseCase\Comment{Base case: violation of (\ref{smallndef})}
		\If{$n$ and $\dst$ violate (\ref{smallndef})} 
		     \State Run the linear tester described in Section \ref{sec:warmup} 
		        and \Return the same answer
		\EndIf
		\EndBaseCase
		\StartMainCase\Comment{General case: (\ref{smallndef}) satisfied}
		  \State  Let $\smash{L=L(n,\dst)=\widetilde{O}(\sqrt{n}/\eps)}$ be as defined in (\ref{uuuu4}) to simplify (\ref{eqsimplified2})

		\For {$j=1,2,\ldots ,\lceil\log 2L\rceil$} \label{eq:hahaha}
		\Comment{Test (\ref{eqsimplified2}), via bucketing}

			\State\label{step:sample:restrictions:t} Sample 
			$s_j=8L \log (2L)\cdot 2^{-j}$ restrictions from $\calD_{\sigma}(\p)$			
			\For {every restriction $\brho$ sampled with $|\stars(\brho)|>0$}
			  \State\label{step:run:testmean} Run $\textsc{TestMean}(|\stars(\brho)|,\p_{|\brho},2^{-j} )$ for $r=O(\log (n/\dst))$ times  
			   \State\label{step:hehe1} \Return \reject if the majority of calls return $\reject$ 
			\EndFor
		 
		\EndFor

			      \For {$j=1,2,\ldots,\lceil \log (\new{4}/\dst)\rceil$} \Comment{Test the second part of (\ref{eq222}), recursively} 
			      \State\label{step:sample:restrictions} Sample 
			      $s_j' = (\new{32}/\dst)\log (\new{4}/\dst)\cdot 2^{-j} $  restrictions from $\calD_{\sigma}(\p)$
			      \For {each restriction $\brho$ sampled satisfying $0<|\stars(\brho)|\le 2\sigma n$} 
			           
				       	\State\label{step:rrrrr} Run $\SubCondUni(|\stars(\brho)|,\p_{|\brho},2^{-j} )$ for
					$t=100 \log (16/\dst)$ times
					\State\label{step:hehe2} \Return \reject
					if the majority of calls return \reject
			      \EndFor
			 \EndFor

		\EndMainCase
		\State \Return \accept
  	\end{algorithmic}
	\caption{$\textsc{SubCondUni}( \dims,\p,\dst )$}\label{algo:subconduni}
\end{algorithm}

\begin{proofof}{\cref{theo:subconduni}}
We start with (i) (the completeness) and prove by induction on $n$ that, when $p$ is uniform,
  $\SubCondUni(n,p,\dst)$ returns $\accept$ with probability at least $2/3$.
The base case when $n=1$ is trivial since (\ref{smallndef}) is violated and we just run the linear tester.
Assuming that the statement is true for all dimensions smaller than $n$, we focus on the
  case when  (\ref{smallndef}) is satisfied; the case when it is violated is again trivial.
Since $p$ is uniform, the random restriction  $\p_{|\brho}$ is uniform with probability one.
On the other hand, note that the total number of restrictions $\brho$ we draw in line \ref{step:sample:restrictions:t} is
  $\sum_j s_j=O(L\log L)=\widetilde{O}(\sqrt{n}/\dst^2)$.
As a result, one can set the constant hidden in the choice of $r$ in line \ref{step:run:testmean} to be sufficiently
  large so that $\SubCondUni$ rejects in line \ref{step:hehe1} with probability no larger than $1/6$.
Similarly, by invoking the inductive hypothesis, 
  one can show that $\SubCondUni$ rejects in line \ref{step:hehe2} with probability no larger than $1/6$.
It follows from a union bound that $\SubCondUni$ rejects with probability at most $1/3$ and 
  this finishes the induction step.
  \medskip

We now turn towards establishing (ii) (the soundness), and hereafter assume that $\dtv(\p,\Uniform)\ge \dst$. 
Again we prove by induction on $n$ that $\SubCondUni(n,p,\dst)$ rejects with probability at least
  $2/3$. For the general case in the induction step,
  it follows from our discussion earlier that either (\ref{eqsimplified2}) or (\ref{eq222}) holds
  (and in particular, the right hand side of (\ref{eqsimplified2}) can be replaced by $1/L$).

For the first case, we note that the left hand side of (\ref{eqsimplified2}) is 
  the expectation of a random variable that takes values between $0$ and $1$ (and the right hand side can be replaced by $1/L$).
Thus, with a simple bucketing argument
  there must exist a $j\in [\lceil \log 2L\rceil]$ such that 
$$
\Prx_{\brho\sim \calD_\sigma(\p)}\left[\frac{\|\mu(\p_{|\brho})\|_2}{\sqrt{n}}\ge 2^{-j}\right]\ge 
  \frac{2^{j-1}}{L\lceil \log 2L\rceil} \ge \frac{2^j}{4L\log 2L}.
$$
It follows that one of the restrictions sampled satisfies the condition above with
  probability at least $1-e^{-2}>5/6$. When this happens, 
  {each call to $\textsc{TestMean}$ in line \ref{step:rrrrr} rejects with probability at least $2/3$. As a result $\SubCondUni$ rejects in line \ref{step:hehe1} with probability at least $2/3$.}

For the second case, again by the bucking argument, there exists a $j\in [\lceil \log (\new{4}/\dst)\rceil]$ such that
$$
\Prx_{\brho\sim \calD_\sigma(\p)}\left[\dtv(\p_{|\brho},\Uniform)\ge 2^{-j}\right]\ge 
\frac{\dst 2^j}{\new{4}\lceil \log (\new{4}/\dst)\rceil}\ge \frac{\dst 2^j}{\new{8} \log (\new{4}/\dst) }.
$$
By (\ref{smallndef}),
  the probability of $|\stars(\brho)|>2\sigma n$ is by a Chernoff bound at most $e^{-\sigma n/3}<(\dst/\new{8})^3$.
Thus,  
\begin{equation*} 
  \Prx_{ \brho\sim \calD_\sigma(\p)} \left[\dtv(\p_{|\brho },\Uniform)\ge 2^{-j}\ \text{and}\ 
  0<|\stars(\brho)|\le 2\sigma n\right]\geq  \frac{\dst 2^j}{\new{16} \log (\new{4}/\dst) }.
\end{equation*}
Given our choice of $\new{s_j'}$, the probability that at least one of $\brho$ satisfies 
  $\dtv(\p_{|\brho},\Uniform)\ge 2^{-j}$ is at least $5/6$.
The probability that the majority of calls reject is also at least $5/6$ (for this case
  we just need $t$ to be a large enough constant). 
Thus, $\SubCondUni$ rejects with probability at least $2/3$.

\medskip
Finally, we bound the query complexity of $\SubCondUni$.
Let $\Phi(\dims,\dst)$ denote the complexity of $\SubCondUni(n,p,\dst)$. We will show by induction on $n$ that 
\begin{equation}\label{eq:sample:complexity:induction}
    \Phi(\dims,\dst ) \leq C\cdot \frac{\sqrt{\dims}}{\dst^2}
     \cdot \log^{c} 
    \left(\frac{\dims}{\dst}\right) \end{equation}
for some absolute constants $C,c>0$.
To fix $C$ and $c$, we let $C_1$ and $c_1$ be two constants such that 
  the query complexity of the linear tester running on $n$ and $\dst$ that violate (\ref{smallndef})
  can be bounded by
$$
C_1\cdot \frac{1}{\dst^2}\cdot \log^{c_1}\left(\frac{\new{1}}{\dst}\right).
$$
We also let $C_2$ and $c_2$ be two constants such that the query 
  complexity of $\SubCondUni$ between
  line \ref{eq:hahaha} and \ref{step:hehe1} (the non-recursive part) can be bounded by 
$$
C_2\cdot \frac{\sqrt{n}}{\dst^2}\cdot \log^{c_2}\left(\frac{n}{\dst}\right).
$$
For this, note that the query complexity of this part is
  (recall that $L=\widetilde{O}(\sqrt{n}/\dst)$)
$$
\sum_{j=1}^{\lceil \log 2L\rceil}  \frac{L\log L}{2^j}\cdot
{ \widetilde{O}\left(\max\left\{\frac{2^{2j}}{\sqrt{n}},2^j \right\}
\right)}\cdot \log \left(\frac{n}{\dst}\right)=\widetilde{O}\left(\frac{\sqrt{n}}{\dst^2}\right).
$$
We then set $C\eqdef 2\max(C_1,C_2)$ and $c\eqdef \max(c_1,c_2)$.

To prove (\ref{eq:sample:complexity:induction}), the base case when $n=1$ is trivial.
For the induction step, the special case when  (\ref{smallndef}) is violated  is 
  also trivial.
For the general case, we have (by the choice of $C$ and $c$)$$
\Phi(n,\dst)\le \frac{C}{2}\cdot \frac{\sqrt{n}}{\dst^2}\cdot \log^c\left(\frac{n}{\dst}\right) +\sum_{j=1}^{\lceil \log (\new{4}/\dst)\rceil}
   s_j'\cdot 100 \log \left(\frac{16}{\dst}\right) \cdot \Phi\big(2\sigma n, 2^{-j}\big).
$$
Using the inductive hypothesis (and our choice of $\sigma$), each term in the second sum becomes 
\begin{align*}
 C\cdot 
\frac{{32}}{\dst}\cdot \log \left(\frac{{4}}{\dst}\right)\cdot
100\cdot \log \left(\frac{16}{\dst}\right)\cdot \sqrt{2\sigma n} \cdot 
2^{ j}\cdot \log^c\left(\sigma n 2^{j{+1}}\right) 
 \le \frac{C}{{32}}\cdot  \frac{\sqrt{n}}{\dst}\cdot 2^j\cdot \log^c\left(\frac{n}{\dst}\right),
\end{align*}
{using $C_0 \geq (32^2\cdot 100)^2 \cdot 2$.}
We can then finish the induction by using 
$$
\sum_{j=1}^{\lceil \log ({4}/\dst)\rceil} 2^j< 2^{\lceil \log ({4}/\dst)\rceil+1}\le \frac{{16}}{\dst}.
$$
This concludes the proof of the theorem.
\end{proofof}

\section{Proof of \cref{thm:main-restriction}}

Let $\calS(t)$ be the uniform distribution supported on all subsets of $[n]$ of size $t$. Given a distribution $\p$ supported on $\bool^n$, we let $\calD(t, \p)$ be the distribution supported on restrictions $\{-1,1, *\}^{n}$ given in a similar fashion to that of $\calD_{\sigma}(\p)$, except we use sets of size $t$; we sample $\bS \sim \calS(t)$ and $\bx \sim p$, and we let $\brho_i = *$ if $i \in \bS$ and $\bx_i$ otherwise.  The bulk of the work goes into proving the following lemma. After the statement, we show the lemma implies \cref{thm:main-restriction}.
\begin{lemma}\label{lem:main-tech}
Let $\p$ be a distribution supported on $\bool^n$, $t \in [n-1]$, and denote
\[ \alpha \eqdef \Ex_{\bT \sim \calS(t)}\Big[\dtv (\p_{\ol{\bT}}, \Uniform ) \Big] 
\ge 0. \]
Then,
\begin{align}
\Ex_{\brho \sim \calD(t, \p)}\left[ \normtwo{ \mu(\p_{|\brho}) }\right] + \Ex_{\brho \sim \calD(t+1, \p)}\left[ \normtwo{ \mu(\p_{|\brho})}\right] &\gsim \frac{t}{n} \cdot \frac{\alpha}{\log^2 n \cdot \log(n/\alpha) \cdot \log(1/\alpha)} .\label{eq:restrict-t}
\end{align}
\end{lemma}

\begin{proofof}{\cref{thm:main-restriction} assuming \cref{lem:main-tech}}
We assume that $5\le \sigma n\le n-5$ in the proof; otherwise the statement is trivially satisfied.
Let $\delta=\min(\sigma,1-\sigma)/2>0$.

For each $t\in [n-1]$, we write for the sake of this proof 
\[ \alpha_t \eqdef \Ex_{\bT \sim \calS(t)}\Big[ \dtv (\p_{\ol{\bT}}, \Uniform ) \Big].  \]
Notice that $\brho \sim \calD_{\sigma}(\p)$ may be sampled by drawing $\bk \sim \Bin(n,\sigma)$ and $\brho \sim \calD(\bk, \p)$. We write $\beta_t$ for the probability over $\bk \sim \Bin(n, \sigma)$ that $\bk = t$. 
Let $$B := \Big\{ t \in [n-1] : t/n \in [\sigma-\delta,\sigma+\delta]\Big\}.$$
Then, by Chernoff's bound we have 
  $\sum_{t\in B} \beta_t\ge 1-2\hspace{0.02cm}e^{-\delta n/5}$ and thus,
\begin{align}
 \sum_{t \in B} \beta_t\cdot \alpha_t \geq  \Ex_{\bT \sim \calS_{\sigma}}\Big[ \dtv(\p_{\ol{\bT}}, \Uniform) \Big] - 
  2\hspace{0.02cm}e^{-\delta n/5}. \label{eq:good-c}
\end{align}
Then, we have
\begin{align}
2\cdot \Ex_{\brho \sim \calD_{\sigma}(\p)}\left[ \|\mu(\p_{|\brho})\|_2 \right]
&\geq \sum_{t \in B} \left( \beta_t \cdot \Ex_{\brho \sim \calD(t, \p)}\left[ \|\mu(\p_{|\brho})\|_{2}\right] + \beta_{t+1}\cdot  \Ex_{\brho \sim \calD(t+1, \p)}\left[ \|\mu(\p_{|\brho})\|_{2} \right]\right) \nonumber \\
&\quad \quad \gsim \sum_{t \in B}\beta_t \left( \Ex_{\brho \sim \calD(t, \p)}\left[ \|\mu(\p_{|\brho})\|_{2} \right] + \Ex_{\brho \sim \calD(t+1, \p)}\left[ \|\mu(\p_{|\brho})\|_{2} \right]\right) \label{eq:sigma-bound}\\[0.8ex]
&\quad\quad \gsim \frac{\sigma}{\log^2 n} \cdot \sum_{t \in B} \beta_t \cdot \frac{ \alpha_{t} }{\log (n/\alpha_t) \log(1/\alpha_t)}\label{eq:sigma-bound22}\\[0.8ex]
&\quad\quad \gsim \frac{\sigma}{\polylog(n)} \cdot
\widetilde{\Omega}\left( \Ex_{\bT \sim \calS_{\sigma}}\Big[ \dtv(\p_{\ol{\bT}}, \Uniform)\Big]  - 2\hspace{0.02cm}e^{-\delta n /5} \right). \label{eq:convexity}
\end{align}
In (\ref{eq:sigma-bound}), we used $t/n \in [\sigma-\delta, \sigma+\delta]$, $\delta=\min(\sigma,1-\sigma)/2$ and $\sigma\ge 5/n$ to have
$$
\frac{\beta_{t+1}}{\beta_t}=\frac{n-t}{t+1}\cdot \frac{\sigma}{1-\sigma} 
\ge 
\frac{(1-\sigma)/2}{(3\sigma/2)+(1/n)}\cdot \frac{\sigma}{1-\sigma}\gsim 1.
$$
In (\ref{eq:sigma-bound22}) we applied  \cref{lem:main-tech} on each $t\in B$.
In (\ref{eq:convexity}) we applied Jensen's inequality (as the function $f(a) = a / (\log(n/a)\log(1/a))$ when $a \neq 0$ and $f(0) = 0$ is convex in $[0,1]$); and (\ref{eq:good-c}).

This finishes the proof of \cref{thm:main-restriction}.
\end{proofof}

\subsection{Robust Pisier Inequality}

In this section, we prove the robust version of Pisier's inequality. Robustness here is equivalent to that of \cite{KhotMS18} where we will consider a function $f \colon \bool^n \to \R$, and we will lower bound a functional on the values of the edges of the hypercube after assigning them directions. 

Formally, fix $n \in \N$ and let $\calH$ be the undirected graph over the hypercube $\{-1, 1\}^n$ that consists of undirected edges $\{x, x^{(i)}\}$ with $x\in \{-1, 1\}^n$ and $i\in [n]$. 
For $i \in [n]$, recall that $\lapl[i] f \colon \bool^n \to \R$ is the linear operator given by $$\lapl[i]f(x) = \frac{f(x) - f(x^{(i)})}{2}.$$

Notice that in~\cref{thm:vanilla-pisier}, every edge $\{x, x^{(i)}\} $ in $\calH$ for a fixed value of $\by$ is counted twice in the right-hand side; once for the endpoint $x$ and once for the endpoint $x^{(i)}$. In the robust version, we may arbitrarily choose, for each edge $\{x, x^{(i)}\}$, whether to ``charge''  the edge to $x$ or to $\smash{x^{(i)}}$. 
For this purpose, we consider an orientation $G$ of $\calH$
  (so for each  
  $\{x,x^{(i)}\}$ in $\calH$, $G$ contains either $(x,x^{(i)})$ or $(x^{(i)},x)$),
  and charge an edge $\{x,x^{(i)}\}$ to $x$ if $(x,x^{(i)})$ is in $G$ and 
  to $x^{(i)}$ otherwise.  
For convenience, we will abuse the notation in the rest of the section
  to use the name of a directed graph (such as $G$)
  to denote its edge set as well, since its vertex set is usually clear from the context. 
So we will write $(u,v)\in G$ if $(u,v)$ is a directed edge in $G$.

We are now ready to state the robust version of Pisier's inequality, which generalizes~\cref{thm:pisier-intro}.

\begin{theorem}[Robust Pisier's inequality]\label{thm:robust-pisier}
Let $f\colon \bool^n \to \R$ be a function with 
\begin{equation}\label{zerosum}
\Ex_{\bx \sim \bool^n}\big[f(\bx)\big]=0
\end{equation} 
 and let $G$ be an orientation of $\calH$. 
Then, for any $s \in [1, \infty)$ we have
\begin{align*}
 \left( \Ex_{\bx \sim \bool^n}\left[ \Big| f(\bx)\right|^s\Big]\right)^{1/s} &\lsim \log n \cdot \left( \Ex_{\bx, \by \sim \bool^n}\left[ \hspace{0.06cm}\left| \sum_{\substack{i \in[n]\\  {(\bx,\bx^{(i)})\in G}}} \by_i \bx_i \lapl[i]  f(\bx)  
\right|^s
\hspace{0.08cm}\right]\hspace{0.03cm}\right)^{1/s}. 
\end{align*}
\end{theorem}

The proof itself follows the template of~\cite[Theorem 2]{NaorS02} and checks that the necessary changes still give the desired inequality. Before beginning with the proof, we recall some basic notions of Fourier analysis on the hypercube, and define some elements which appear in the argument. Recall that any function $f \colon \bool^n \to \R$ has a unique Fourier expansion \smash{$f(x) = \sum_{S \subset [n]} \hat{f}(S) \chi_{S}(x)$}, where \smash{$\chi_{S}(x) = \prod_{i \in S} x_i$} are the Fourier characters, and \smash{$\hat{f}(S) = \Ex_{\bx \sim \bool^n}[f(\bx) \chi_{S}(x)]$}. For $\rho > 0$ and $x \in \bool^n$, we let $N_{\sigma}(x)$ be the distribution supported on $\bool^n$ given sampling $\by \sim N_{\sigma}(x)$, where for each $i \in [n]$, we set $\by_i = \bx_i$ with probability $1-\sigma$, and a uniform random bit with probability $\sigma$. We denote $T_{\rho} f(x) = \Ex_{\by \sim N_{\rho}(x)}[f(\by)]$, and we have for every $S \subset [n]$, $\hat{T_{\rho} f}(S) = \rho^{|S|} \hat{f}(S)$. In a slight abuse of notation, we consider for any $x, y \in \bool^n$ and $t \in [0,1]$, the distribution $N_{t,1-t}(x, y)$, supported on $\bool^n$, to be the distribution given by letting $\bz \sim N_{t, 1-t}(x, y)$ have each $i \in [n]$ set to $\bz_i = x_i$ with probability $t$ and $\bz_i = y_i$ otherwise. For a function $g \colon \bool^n \to \R$ and $t \in [0,1]$, the function $g \colon \bool^n \times \bool^n\to \R$ is given by letting $g_{t, 1-t}(x, y) = \Ex_{\bz \sim N_{t, 1-t}(x, y)}[g(\bz)] = \sum_{S \subset [n]} \hat{g}(S) \prod_{i \in S}(tx_i + (1-t) y_i)$. Lastly, for any $\gamma > 0$, we let $\Delta^{\gamma} f$ be the linear operator given by $\Delta^{\gamma} f(x) = \sum_{S \subset [n]} \hat{f}(S) |S|^{\gamma} \chi_{S}(x)$.\medskip

\begin{proofof}{\cref{thm:robust-pisier}}
Let $\rho = 1 - 1/(n+1)$ and $q \in [1, \infty]$ such that $\frac{1}{s} + \frac{1}{q} = 1$. Fix the function $g \colon \bool^n \to \R$ with $\|g\|_{q} = 1$ satisfying $\langle T_{\rho} f, g \rangle = \|T_{\rho} f\|_{s}$. We have
\begin{align}
\rho^n \|f\|_s &\leq \|T_{\rho}f\|_s = \langle T_{\rho} f, g \rangle = \sum_{\substack{S \subset [n] \\ S \neq \emptyset}} \rho^{|S|} \hat{f}(S) \hat{g}(S) \nonumber \\ 
&= \frac{1}{\Gamma(1+\gamma)} \int_{0}^{\rho} \left( \sum_{\substack{S \subset [n] \\ S \neq \emptyset}} t^{|S| - 1} |S|^{\gamma + 1} \hat{f}(S) \hat{g}(S) \right) \left( \log(\rho/t) \right)^{\gamma} dt \label{eq:pisier-2}\\
&= \frac{1}{\Gamma(1+\gamma)} \int_{0}^{\rho} \frac{1}{1-t} \Ex_{\bx , \by \sim \bool^n}\left[ g_{t, 1-t}(\bx, \by) \sum_{i=1}^n \by_i \bx_i \lapl[i] \Delta^{\gamma} f(\bx)\right] \left( \log(\rho/t)\right)^{\gamma} dt, \label{eq:pisier-3}
\end{align}
where \eqref{eq:pisier-2} follows from writing $\rho^{|S|} = \frac{1}{\Gamma(1+\gamma)} \int_{0}^{\rho} t^{|S|-1} |S|^{\gamma+1} \log(\rho/t)^{\gamma} dt$, and \eqref{eq:pisier-3} follows from the Fourier expansion of $g_{t, 1-t}(x, y) \sum_{i=1}^n y_i x_i \lapl[i] \Delta^{\gamma} f(x)$. The main step is obtaining \eqref{eq:pisier-4} below:
\begin{align}
\Ex_{\bx, \by \sim \bool^n}\left[ g_{t, 1-t}(\bx, \by) \sum_{i=1}^n \by_i \bx_i \lapl[i] \Delta^{\gamma} f(\bx)\right] &= 2\Ex_{\bx, \by \sim \bool^n}\left[ g_{t, 1-t}(\bx, \by)\sum_{\substack{i \in [n] \\ (\bx, \bx^{(i)})\in G}} \by_i \bx_i \lapl[i] \Delta^{\gamma} f(\bx) \right] \label{eq:pisier-4} \\
&\leq 2\left(\Ex_{\bx , \by \sim \bool^n}\left[ \left| \sum_{\substack{i\in[n] \\ (\bx, \bx^{(i)})\in G}} \by_i \bx_i \lapl[i] \Delta^{\gamma} f(\bx) \right|^s\right]\right)^{1/s} , \label{eq:pisier-5}
\end{align}
since \eqref{eq:pisier-5} follows from \eqref{eq:pisier-4} by noting that $\|g_{t, 1-t}\|_{q} \leq 1$. We may now proceed substituting \eqref{eq:pisier-5} into \eqref{eq:pisier-2}. For $\rho = 1 - 1/(n+1)$, we have that for $\gamma$ approaching $0$, we have
\begin{align*}
\frac{1}{\Gamma(1+\gamma)} \int_{0}^{\rho} \frac{\log(\rho/t)^{\gamma}}{1-t} dt \lsim \log n,
\end{align*}
so we obtain 
\begin{align*}
\rho^{n} \|f\|_s \lsim \log n \left(\Ex_{\bx , \by \sim \bool^n}\left[ \left| \sum_{\substack{i\in[n] \\ (\bx,\bx^{(i)})\in G}} \by_i \bx_i \lapl[i] \Delta^{\gamma} f(\bx) \right|^s\right]\right)^{1/s} ,
\end{align*}
and the right-hand side approaches the desired quantity, while the left-hand side is independent of $\gamma$. It thus remains to show \eqref{eq:pisier-4}. For simplicity, let $h_i = \lapl[i] \Delta^{\gamma} f$. We consider summing the edges of $\calH$ according to the orientation of the edge:
\begin{align}
&\Ex_{\bx, \by \sim \bool^n}\left[ g_{t, 1-t}(\bx, \by) \sum_{i=1}^n \by_i \bx_i h_i(\bx)\right] \label{eq:pisier-6}\\
&\qquad\qquad= \Ex_{\bx\sim\bool^n}\left[ \sum_{\substack{i \in [n] \\ (\bx,\bx^{(i)})\in G}} \Ex_{\by \sim \bool^n}\left[ \by_i \left( g_{t, 1-t}(\bx, \by) h_i(\bx) - g_{t, 1-t}(\bx^{(i)}, \by) h_i(\bx^{(i)})\right) \right]  \right] \nonumber
\end{align}
Now, for a fixed $x$ and $i$, we have
\begin{align*}
&\Ex_{\by \sim \bool^n}\left[ \by_i \left(g_{t, 1-t}(x, \by) h_i(x) - g_{t, 1-t}(x^{(i)}, \by) h_i(x^{(i)}) \right)\right] \\ &\hspace{3cm}= \Ex_{\by \sim \bool^n}\left[ \by_i g_{t,1-t}(x, \by)\left(h_i(x) - h_i(x^{(i)})\right) \right].
\end{align*}
This is because when expanding the terms in $g_{t, 1-t}(x, \by) = \Ex_{\bz\sim N_{t,1-t}(x, \by)}[g(\bz)]$ two things may happen. 1) Either $\bz_i = \by_i$ (for $(1-t)$-fraction of terms), in which case the terms in $\Ex_{\bz \sim N_{t,1-t}(x, \by)}[g(\bz)]$ and $\Ex_{\bz \sim N_{t, 1-t}(x^{(i)}, \by)}[g(\bz)]$ are the same, and these are scaled by $h_i(x)$ and $-h_i(x^{(i)})$, respectively. Or, 2) $\bz_i = x^{(i)}$, in which case terms in $\Ex_{\bz \sim N_{t,1-t}(x, \by)}[g(\bz)]$ and $\Ex_{\bz\sim N_{t, 1-t}(x^{(i)}, \by)}[g(\bz)]$ are scaled by $h_i(x)$ and $-h_i(x^{(i)})$, respectively. However, in this case, these terms are all \emph{independent} of $\by_i$, and thus are added and subtracted for an overall contribution of zero. With these considerations, \eqref{eq:pisier-6} becomes
\begin{align*}
\Ex_{\bx, \by \sim \bool^n}\left[ g_{t, 1-t}(\bx, \by)\sum_{i=1}^n \by_i\bx_i h_i(\bx)\right] &= \Ex_{\bx, \by \sim \bool^n}\left[ \sum_{\substack{i \in [n] \\ (\bx,\bx^{(i)})\in G}} \by_i g_{t, 1-t}(\bx, \by) \left(h_i(\bx) - h_i(\bx^{(i)})\right) \right]
\end{align*}
and since $h_i(\bx) - h_i(\bx^{(i)}) = 2\Delta^{\gamma} \lapl[i] f(\bx)$, we obtain the desired bound.
\end{proofof}

\subsection{Plan of the Proof of Lemma \ref{lem:main-tech}}

Let $t\in [n-1]$ be the parameter in the statement of Lemma \ref{lem:main-tech}.
For clarity, the rest of this section will always use $T$ to denote a size-$t$ subset of $[n]$ and $S$ to denote a size-$(t+1)$ subset of $[n]$.
For each size-$t$ subset $T$ of $[n]$, we write 
  $\alpha(T) \eqdef \dtv(\p_{\ol{T}}, \Uniform).$
Then $\alpha = \Ex_{\bT \sim \calS(t)}[\alpha(\bT)]$.
We also write $\smash{\calH( {T})}$ to denote the undirected graph over 
  the hypercube $\smash{\{-1,1\}^{\ol{T}}}$ that consists of undirected edges $\smash{\{x , x^{(i)}\}}$ with $x\in \{-1,1\}^{\ol{T}}$ and $i\in \ol{T}$.
Again we will abuse the notation in the rest of the section to 
  refer to  $\calH( {T})$ as its edge set as well.  

The proof of Lemma \ref{lem:main-tech} consists of two steps. 
For each $t$-subset $T$, we first classify undirected edges $\{x,x^{(i)}\}$ in $\calH( {T})$ into 
  different types according its \emph{weight} defined as 
\[   w\big(\{x, x^{(i)}\}\big) \eqdef \dfrac{\big| p_{\ol{T}}(x) - p_{\ol{T}}(x^{(i)}) \big|}{\max\left\{ p_{\ol{T}}(x), p_{\ol{T}}(x^{(i)})\right\}}.\]
For each type of edges in $\calH( {T})$, we describe a method in Section \ref{sec:color} to assign 
  each edge a direction.
This then leads to a sequence of directed graphs over $\smash{\{-1,1\}^{\ol{T}}}$,
  one for each type of edges in $\smash{\calH( {T})}$,
  and their union is an orientation $G$ of $\calH({T})$ over $\smash{\{-1,1\}^{\ol{T}}}$.
At the end of Section \ref{sec:color} we apply the robust Pisier inequality (Theorem~\ref{thm:robust-pisier}) on 
  a shifted and scaled version of the probability mass function of $p_{\overline{T}}$  
   with the orientation $G$ and $s=1$. The result will be Lemma~\ref{lem:bucket},
 which says there is a type of edges in $\calH( {T})$
  such that its corresponding directed graph has
\begin{equation}\label{haha5}
\Ex_{\bx\sim p_{\overline{T}}} \left[\sqrt{\text{out-degree of $\bx$}}\right]
\end{equation}
bounded from below by a quantity that is linear in $\alpha(T)$; see Lemma \ref{lem:bucket} for details.

In the second step, we use this family of directed graphs promised by Lemma \ref{lem:bucket},
  one for each $t$-subset $T$, with the desired bound on   
  (\ref{haha5}) to finish the proof of Lemma \ref{lem:main-tech}.
To this end we first apply standard bucketing arguments in Section \ref{sec:bucketing}
  to simplify the situation, by focusing on one specific type of directed edges
  that makes the most significant contribution in the family. 
The final connection from these directed graphs to the mean vectors of randomly restricted distributions
  is made in Sections \ref{sec:graphs-to-mean}, \ref{sec:case1}, and \ref{sec:case2}. 
There will be two cases,
  depending on whether the type of edges we consider has large or small weights.
They are handled in Sections \ref{sec:case1} and \ref{sec:case2}, respectively.
 
\subsection{From Total Variation to Directed Graphs}\label{sec:color}

\def\outdegree{\mathrm{outdeg}}

Let $\ell$ be a probability distribution over $\{-1,1\}^m$ with $m=n-t$.
(Later we will identify $\ell$ as $p_{\overline{T}}$ for some $t$-subset $T$ of $[n]$,  and 
  $\{-1,1\}^m$ as $\smash{\{-1,1\}^{\overline{T}}}$.)
Let $\calH$  denote the undirected graph over $\{-1,1\}^m$
    that consists of undirected edges $\{x,x^{(i)}\}$ for all $x\in \{-1,1\}^m$ and $i\in [m]$.
Looking ahead, the purpose of this section is to construct
  an orientation $G$ of $\calH$ 
  for our application of Theorem \ref{thm:robust-pisier} later with $s=1$ and the function  
  $f \colon \{-1,1\}^{m} \to [-1, \infty)$ given by:
\begin{align}
 f(y) = 2^{m} \cdot \ell(y) - 1.  \label{eq:bias-function} 
 \end{align}
Note that $\Ex_{\by}[f(\by)]=0$ and the left hand side of the robust Pisier inequality is
  $2\hspace{0.02cm}\dtv(\ell,\Uniform)$.

For the construction of $G$,
  we start with a classification of undirected edges in $\calH$.

\begin{definition}
An undirected edge $\{x, x^{(i)}\} \in \calH$ is said to be a \emph{zero} edge if 
  $ \ell(x)=\ell(x^{(i)})$ (they are called zero edges 
  because the difference $\ell(x)-\ell(x^{(i)})=0$).

For each nonzero edge $\{x,x^{(i)}\}\in \calH$, we define its \emph{weight} as
$$
  w(\{x, x^{(i)}\}) \eqdef \dfrac{\big| \ell(x) - \ell(x^{(i)}) \big|}{\max\left\{ \ell(x), \ell(x^{(i)})\right\}}.
$$
Note that the weight 
  of an undirected edge is always in $(0,1]$.
We say a nonzero edge is \emph{uneven} if its weight is at least $2/3$; otherwise, we call
  it an \emph{even} edge (i.e., any nonzero edge with weight smaller than $2/3$).
We say an even edge is at \emph{scale} $\kappa$ for some $\kappa\ge 1$ if 
\[ 2^{-\kappa }< w\big(\{x,x^{(i)}\}\big)  \leq 2^{-\kappa+1 }. \]

We write {$\calH^{[z]}$ to denote the set of all zero edges,} $\calH^{[u]}$ to denote the set of all uneven edges, and $\calH^{[\kappa]}$ for each $\kappa\ge 1$ to denote the set of even edges at scale $\kappa$.
Hence $\calH^{[z]}$, $\calH^{[u]}$ and $\calH^{[\kappa]}$ with $\kappa\ge 1$ together form a
  partition of $\calH$; we also view them as undirected graphs over $\{-1,1\}^m$.
\end{definition}

Next we construct a sequence of directed graphs $G^{[z]}$, $G^{[u]}$ and $G^{[\kappa]}$, $\kappa\ge 1$, 
  as orientations of $\calH^{[z]}$, $\calH^{[u]}$ and $\calH^{[\kappa]}$, respectively.
We start with $G^{[z]}$ and $G^{[u]}$.
For each zero edge $\{x,x^{(i)}\}\in \calH^{[z]}$, 
  we orient it arbitrarily in $G^{[z]}$.~Next for each uneven edge $\{x,x^{(i)}\}\in \calH^{[u]}$, we orient  
  it from $x$ to $x^{(i)}$ if $\ell(x)>\ell(x^{(i)})$ (note that if $\ell(x)=\ell(x^{(i)})$ then it is 
  a zero edge).

Orientations of even edges at scale $\kappa$ in $G^{[\kappa]}$ are more involved.
For a fixed $\kappa\ge 1$, we consider 
  $\calH^{[\kappa]}$ as an undirected graph over $\{-1,1\}^m$.
We will consider a bijection $\varrho_{\kappa} \colon \bool^{m} \to [2^m]$ as an ordering
  of vertices in $\bool^m$ (so $x$ is the $\varrho_{\kappa}(x)$-th vertex in the ordering) such that the following property holds:
For every $i\in [2^m-1]$, the degree of $\varrho_{\kappa}^{-1}(i)$ has the largest 
  degree among all vertices 
  in the subgraph of $\calH^{[\kappa]}$ induced by $\{\varrho_{\kappa}^{-1}(j):j\ge i\}$.
Such a bijection exists, e.g., by keeping deleting vertices one by one and each time 
  deleting the one with the largest degree in the remaining graph (with tie breaking done arbitrarily).
We fix such a bijection $\varrho_{\kappa}$ and use it to orient edges in $G^{[\kappa]}$ as follows:
  For each undirected $\{x,x^{(i)}\}\in \calH^{[\kappa]}$, we orient it from $x$ to $x^{(i)}$ if 
  $\varrho_{\kappa}(x)<\varrho_{\kappa}(x^{(i)})$ and orient it from $x^{(i)}$ to $x$ otherwise.
As a result, every $(x,x^{(i)})\in G^{[\kappa]}$ has $\varrho_{\kappa}(x)<\varrho_{\kappa}(x^{(i)})$.    

{The above orientation will effectively streamline an argument from \cite[Section 6]{KhotMS18}.} We record the property needed later for the orientation $G^{[\kappa]}$ of $\calH^{[\kappa]}$:

\begin{lemma}\label{lem:greedy}
Let $U$ be a set of vertices in $\{-1,1\}^m$ and let $v\in \{-1,1\}^m\setminus U$.
If the out-degree of every vertex $u\in U$ in $G^{[\kappa]}$ is bounded from above by a positive integer $g$,
  then the number of directed edges $(u,v)$ from a vertex $u\in U$ to $v$ in $G^{[\kappa]}$ is also at most $g$.
\end{lemma}  
\begin{proof}
Consider the vertex with the smallest $\varrho_\kappa(\cdot)$ value in $U\cup \{v\}$.
If it is $v$, then every undirected $\{u,v\}$ with $u\in U$,
  if any, in $\calH^{[\kappa]}$ is oriented as $(v,u)$ in $G^{[\kappa]}$. 
So the number we care about is $0$.

Otherwise, let $u\in U$ be the vertex with the smallest $\varrho_{\kappa}(\cdot)$ value among $U\cup \{v\}$.
Then at the time when $u$ is picked, all vertices $U\cup \{v\}$ remain in current undirected subgraph of $\calH^{[\kappa]}$, denoted by $H$.
At this moment, the degree of $u$ in $H$ is exactly its out-degree in $G^{[\kappa]}$,
  which by assumption is at most $g$.
On the other hand, by the choice of $u$, $v$ has degree at most $g$ in $H$.
Since the whole set $U$ remains in $H$,
  the number of undirected edges $\{u,v\}$, $u\in U$, in $\calH^{[\kappa]}$ is at most $g$. 
Even if all of them are oriented towards $v$ in $G^{[\kappa]}$, the number we care about in the lemma
  is at most $g$.
\end{proof}\medskip

With $G^{[z]}$, $G^{[u]}$ and $G^{[\kappa]}$ ready, we finally define $G$ to be 
  the union of these graphs, which is an orientation of $\calH$ over $\{0,1\}^m$.
Applying the robust Pisier inequality on $f$, $G$ and $s=1$, we have
\begin{align}
\frac{\dtv(\ell,\Uniform)}{\log n} &\lsim 
 \Ex_{\bx, \by \sim \bool^m}\left[ \hspace{0.06cm}\left| \sum_{\substack{i \in [m] \\ (\bx,\bx^{(i)})\in G}} \by_i \bx_i \lapl[i]  f(\bx)\right|\hspace{0.06cm}\right]  
				  \leq \Ex_{\bx \sim \{-1,1\}^{m}}\left[ \sqrt{\sum_{\substack{i\in [m]\\ (\bx,\bx^{(i)})\in G}} \big(\lapl[i] f(\bx) \big)^2  } \right].
\end{align}
The second inequality is Khintchine's, which implies that for any vector $a \in \R^m$, $$\Ex_{\by\sim\{-1,1\}^m}\left[\hspace{0.05cm} \left|\sum_{i\in [m]} \by_i a_i\right|\hspace{0.05cm}\right] \leq \sqrt{\sum_{i\in [m]} a_i^2}.$$ 
Letting $G'$ be the directed graph that contains the union of edges in $G^{[u]}$ and $G^{[\kappa]}$,
  $\kappa\in \mathbb{Z}_{\ge 0}$, but not those in $G^{[z]}$, we can continue the inequality above to have
$$
\hspace{-0.2cm}\Ex_{\bx \sim \{-1,1\}^{m}}\left[ \sqrt{\sum_{\substack{i\in [m]\\ (\bx,\bx^{(i)})\in G}} \big(\lapl[i] f(\bx) \big)^2  } \right] 
= \Ex_{\bx \sim \ell}\left[ \sqrt{\sum_{\substack{i\in [m] \\ (\bx,\bx^{(i)})\in G'}} \left(\frac{\lapl[i]  f(\bx)}{1 + f(\bx)} \right)^2 }\right]
= \Ex_{\bx \sim \ell}\left[ \sqrt{\sum_{\substack{i\in [m] \\ (\bx,\bx^{(i)})\in G'}} \left(\frac{\lapl[i]  \ell(\bx)}{\ell(\bx)} \right)^2 }\right].
$$
For the first equation we note that zero edges do not contribute anything and utilize
  importance sampling, by noting that $\ell(x) = (1+f(x))/2^m$.
Also note we never run into a situation of $0/0$ in the second expectation because
  if $(x,x^{(i)})\in G'$ has $\ell(x)=0$, then either $\ell(x^{(i)})=0$ and it is a zero edge
  that should have been excluded from $G'$, or $\ell(x^{(i)})>0$ and $\{x,x^{(i)}\}$ is uneven.
Then by the construction of $G^{[u]}$ we have $(x^{(i)},x)\in G$ instead of $(x,x^{(i)})$.

The next lemma connects the sum for each $x\in \{-1,1\}^m$ in the last expectation
  with its out-degrees in the directed graphs $G^{[u]}$ and $G^{[\kappa]}$ constructed.

\begin{lemma}\label{lem:uneven-to-degree}
For every $x \in \{-1,1\}^m$, 
\begin{align*} 
\sum_{\substack{i\in [m]\\ (x,x^{(i)})\in G'}} \left(\frac{\ell(x)-\ell(x^{(i)})}{\ell(x)}\right)^2
\le \outdegree\big(x,G^{[u]}\big)+\sum_{\kappa\ge 1} 2^{-2\kappa+6}\cdot \outdegree\big(x,G^{[\kappa]}\big).
\end{align*}
\end{lemma}

\begin{proof}
First note that each edge $(x,x^{(i)})\in G'$ is nonzero and   
  either lies in $G^{[u]}$ or $G^{[\kappa]}$ for some $\kappa\ge 1$.
If $(x,x^{(i)})$ is in $G^{[u]}$, then by the way we orient edges in $G^{[u]}$,
  we have $\ell(x)>\ell(x^{(i)})$ and 
  this implies that the contribution of each such edge to the sum is at most $1$.

Next assume that $(x,x^{(i)})\in G^{[\kappa]}$ for some $\kappa\ge 1$.
Since $\{x,x^{(i)}\}$ is even, we have   $\ell(x),\ell(x^{(i)})>0$ 
  (otherwise it is either zero or uneven).
Using $w(\{x,x^{(i)}\})<2/3$ (otherwise it is uneven), we have 
$$
\frac{\max\{\ell(x),\ell(x^{(i)})\}}{\min\{\ell(x),\ell(x^{(i)})\}}\le 3.
$$
As a result, we have
$$
\frac{|\ell(x) - \ell(x^{(i)})|}{\ell(x)}\le
w(\{x, x^{(i)}\}) \cdot \frac{\max\{ \ell(x), \ell(x^{(i)})\}}{\min\{ \ell(x) , \ell(x^{(i)})\}} \leq 3 w(\{x, x^{(i)}\}) \leq 2^{-\kappa + 3}.
$$
So the contribution of each such edge is at most $2^{-2\kappa+6}$ and we
  thus obtain the desired bound.
\end{proof}

We are now ready to prove the main lemma of this subsection:
\begin{lemma}\label{lem:bucket}
Letting $\beta=\dtv(\ell,\Uniform)$, one of the following two conditions must hold:\vspace{0.1cm}
\begin{itemize}
\item Either the directed graph $G^{[u]}$ of uneven edges satisfies:
\begin{align*}
\Ex_{\bx \sim \ell}\left[ \sqrt{\outdegree (\bx, G^{[u]})} \right]  \gsim \frac{\beta}{\log n} 
\end{align*}
\item Or, there exists a $\kappa \in [O(\log(n/\beta))]$ such that the directed graph $G^{(\kappa)}$ satisfies:
\begin{align*}
\Ex_{\bx \sim \ell}\left[ \sqrt{\outdegree(\bx, G^{(\kappa)} )}\right] &\gsim \frac{2^{\kappa} \cdot \beta }{\log n\cdot  \log(n/\beta )}. 
\end{align*}
\end{itemize}
\end{lemma}

\begin{proof}
It follows from Lemma \ref{lem:uneven-to-degree} that
\begin{align}
\nonumber\frac{\beta}{\log n}&\lsim \Ex_{\bx\sim \ell}\left[\sqrt{\outdegree(\bx,G^{[u]})+\sum_{\kappa\ge 1}
  2^{-2\kappa+2}\cdot \outdegree(\bx,G^{[\kappa]})}\right]\\
  &\le \Ex_{\bx\sim \ell}\left[\sqrt{\outdegree(\bx,G^{[u]})}\right]+\sum_{\kappa=1}^{O(\log (n/\beta))}
  2^{-\kappa+1}\cdot \Ex_{\bx\sim \ell}\left[ \sqrt{\outdegree(\bx,G^{[\kappa]})}\right]
  +o\left(\frac{\beta}{\log n}\right).\label{eqeqeq}
\end{align}
In (\ref{eqeqeq}) we used the concavity of $\sqrt{\cdot}$ as well as the fact that the degrees are always bounded by $n$.
Lemma \ref{lem:bucket} then follows from (\ref{eqeqeq}).
\end{proof}

\subsection{Bucketing}\label{sec:bucketing}

We now start the proof of Lemma \ref{lem:main-tech}.
For each $t$-subset $T$ of $[n]$, let
$\alpha(T)=\dtv(p_{\ol{T}},\Uniform)$ and thus, $\alpha=\Ex_{\bT\sim \calS(t)}[\alpha(T)]$. 
For each $T$, we partition undirected edges in $\calH( {T})$ into 
  $\smash{\calH^{[z]}( {T})}$ (zero edges), $\smash{\calH^{[u]}( {T})}$ (uneven edges), and 
  $\smash{\calH^{[\kappa]}( {T})}$ (even edges at scale $\kappa\ge 1$).
We orient these edges to obtain directed graphs $\smash{G^{[u]}( {T})}$ and $\smash{G^{[\kappa]}( {T})}$.
We apply Lemma \ref{lem:bucket} on $p_{\ol{T}}$ to conclude that one of the two conditions 
  holds for either  $\smash{G^{[u]}(T)}$ or one of the graphs $\smash{G^{[\kappa]}(T)}$, $\kappa\in [O(\log (n /\alpha(T)))]$.

Since $\alpha(T) \in [0, 1]$, there exists a  
$\zeta>0$ such that with probability at least $\zeta$ over $\bT \sim \calS(t)$, 
\[ \alpha(\bT) \gsim \frac{\alpha}{\zeta \log(1/\alpha)}. \]
Therefore, via another bucketing argument and Lemma~\ref{lem:bucket}, there exist two cases:
\begin{flushleft}\begin{itemize}
\item \textbf{Case 1}: With probability at least $\zeta / 2$ over the draw of $\bT \sim \calS(t)$, we have that the directed graph $G^{[u]}(\bT)$ of {uneven edges} of $p_{\ol{\bT}}$ over 
  $\smash{\{-1,1\}^{\ol{\bT}}}$ satisfies
\[ \Ex_{\bx \sim \p_{\ol{\bT}}}\left[ \sqrt{\outdegree\big(\bx, G^{[u]}(\bT)\big)} \right] \gsim \frac{\alpha}{\zeta \log n \log(1/\alpha)}.  \]
Since the out-degree is always between $0$ and $n$, there exist two parameters $d \in [n]$ and $\xi>0$
  such that with probability $\zeta / (2 \log n)$ over the draw of $\bT \sim \calS(t)$,  we have
\[ \Prx_{\bx \sim \p_{\ol{\bT}}}\left[ d \leq \outdegree(\bx, G^{[u]}(\bT)) \leq 2d \right] \geq \xi \]
and $\xi$ satisfies
\begin{align}
\sqrt{d} \cdot \xi &\gsim \frac{\alpha}{\zeta \log^2 n \log(1/\alpha)} .\label{eq:case-1-params}
\end{align}
\item \textbf{Case 2}: There exists a parameter $\kappa \in [O(\log(n/\alpha))]$ (using $\zeta\le 1$) such that with probability at least $\smash{\zeta / (2 \cdot O(\log(n/\alpha)))}$ over the draw of $\bT \sim \calS(t)$, the directed graph $G^{[\kappa]}(\bT)$ of {even edges} at scale $\kappa$ of
  $p_{\ol{\bT}}$ over $\smash{\{-1,1\}^{\ol{\bT}}}$ satisfies
\[ \Ex_{\bx \sim \p_{\ol{\bT}}}\left[ \sqrt{\outdegree(\bx, G^{[\kappa]}(\bT))}\right] \gsim \frac{\alpha \cdot 2^{\kappa}}{\zeta \log n \log(n/\alpha)\log(1/\alpha)}. \]
By a bucketing argument again, there exist parameters $d \in [n]$ and $\xi >0$
  such that with probability at least $\zeta / (2 \cdot \log n \cdot O(\log(n/\alpha)))$ over the draw of $\bT \sim \calS(t)$,  we have
\[ \Prx_{\bx \sim \p_{\ol{\bT}}}\left[ d \leq \outdegree(\bx, G^{[u]}(\bT)) \leq 2d \right] \geq \xi \]
and $\xi$ satisfies
\begin{align}
\sqrt{d} \cdot \xi \gsim \frac{\alpha \cdot 2^{\kappa}}{\zeta \log^2 n \log(n/\alpha)\log(1/\alpha)}. \label{eq:case-2-params}
\end{align}
\end{itemize}\end{flushleft} 
\subsection{From Directed Graphs to Mean Vectors}\label{sec:graphs-to-mean}

In this section, we will show the crucial connection between analyzing the family of graphs defined in Section~\ref{sec:color} and \ref{sec:bucketing} and mean vectors of restrictions of the distribution. Consider a fixed distribution $\p$ supported on $\{-1,1\}^n$ and let $t \in [n-1]$. We consider the family of directed graphs $G^{[u]}(T)$ and $G^{[\kappa]}(T)$, $\kappa\ge 1$, for each $t$-subset $T$ of $[n]$.

It will be convenient to represent directed edges of these directed graphs
  as $(y,i)  \in \{-1,1\}^{\ol{T}} \times \ol{T}$: we say
  $(y,i)$ is in a graph if it is the case for $(y,y^{(i)})$. 
Let $\pi=(\pi(1),\ldots,\pi(t+1))$ be an (ordered) sequence
  of $t+1$ distinct indices from $[n]$.
We use $S(\pi)$ to denote the corresponding $(t+1)$-subset $\{\pi(1),\ldots,\pi(t+1)\}$.
Given $\pi$ and $y\in \{-1,1\}^n$, we define a restriction $\rho(\pi,y)\in \{-1,1,*\}^n$ as 
\[ \rho(\pi,y)_{i} = \left\{\begin{array}{ll} * & i=\pi(j)\ \text{for some $j\in [t+1]$} \\[0.5ex]
							y_i & \text{otherwise} \end{array} \right. .\]
We will also consider sequences $\tau=(\tau(1),\ldots,\tau(t))$ of $t$ distinct indices
  from $[n]$; its corresponding set $S(\tau)$ and the restriction $\rho(\tau,y)$ given $y\in \{-1,1\}^n$
  are defined similarly.							
							
We consider a slightly different but equivalent way of 
  drawing $\brho$ from $\calD(t+1,\p)$ and $\brho'$ from $\calD(t,\p)$ which we will use to analyze Case 1 and Case 2 as specified in Section~\ref{sec:bucketing}.
We consider sampling $\brho \sim \calD(t+1, \p)$ according to the following procedure:
\begin{flushleft}\begin{enumerate}
\item First, sample a sequence of $t+1$ random indices $\bpi=(\bpi(1),\ldots,\bpi(t+1))$ uniformly from $[n]$
  without replacements (so the set $S(\bpi)$ can be viewed equivalently as drawn from $\calS(t+1)$).
\item Then, sample $\by \sim \p$.
\item Finally, set $\brho=\rho(\bpi,\by)$.
\end{enumerate}\end{flushleft}
Similarly we consider sampling $\brho'\sim \calD(t,\p)$ according to the following procedure:
\begin{flushleft}\begin{enumerate}
\item First, sample a sequence of $t$ random indices $\btau=(\btau(1),\ldots,\btau(t ))$ uniformly from $[n]$
  without replacements (so the set $S(\btau)$ can be viewed equivalently as drawn from $\calS(t )$).
\item Then, sample $\by \sim \p$.
\item Finally, set $\brho'=\rho(\btau,\by)$.
\end{enumerate}\end{flushleft}
A useful observation is that for each $i\in [t+1]$, the distribution of $(\bpi_{-i},\by)$ is the 
  same as that of $(\btau,\by)$, where $\bpi_{-i}$ denotes the $t$-sequence obtained from $\bpi$
  after removing 
  its $i$th entry. As a result, $\rho (\bpi_{-i},\by)$ is distributed according to $\calD(t,\p)$.

Next we state the lemma that gives the connection between graphs and mean vectors.

\begin{lemma}\label{lem:graph-to-mean}
Let $\pi$ be a $(t+1)$-sequence of distinct elements in $[n]$ and $y \in \{-1,1\}^n$. Then~for~every $i \in [t+1]$, we have
\begin{align}
 \left|\hspace{0.02cm}\mu\big(\p_{|\rho(\pi, y)}\big)_{\pi(i)}\right| &\geq \frac{1}{3} \cdot \ind\left\{ \big(y_{\ol{S(\pi_{-i})}}, \pi(i)\big) \in G^{[u]}( {S(\pi_{-i})})  \right\} \nonumber \\[0.4ex]
 	&\qquad\qquad  + \sum_{k \ge 1} 2^{- \kappa-1}  \cdot \ind\left\{ \big(y_{\ol{S(\pi_{-i})}}, \pi(i)\big) \in G^{[\kappa]}( {S(\pi_{-i})})  \right\} . \label{eq:graph-to-mean}
\end{align}
\end{lemma}

\begin{proof}
First recall that $\smash{G^{[u]}( {S(\pi_{-i})})}$ and $\smash{G^{[\kappa]}( {S(\pi_{-i})})}$ 
  are orientations of disjoint undirected edges. 
So the right-hand side of (\ref{eq:graph-to-mean}) is non-zero for at most one value. 

Let $\ell=p_{|\rho(\pi,y)}$, $T=S(\pi_{-i})$, $z = y_{\ol{T}}$ and $z' = z^{(\pi(i))}$. Writing 
\[ a = \Prx_{\bx \sim \p}\left[\bx_{\ol{T}} = z \right] \ \quad \text{and}\ \quad a' = \Prx_{\bx \sim \p}\left[ \bx_{\ol{T}} = z' \right],\]
we have $|\mu(\ell)_{\pi(i)}| = |a - a' | / (a + a')$. The weight $w( \{ z, z'\})$,   defined as $|a - a'| / \max\{a, a'\}$, is at most $2\cdot|\mu(\ell)_{\pi(i)}|$. If $\smash{(z, \pi(i)) \in G^{[u]}(T)}$ is uneven, then the weight is at least $2/3$ and thus,~$|\mu(\ell)_{\pi(i)}| \geq 1/3$. If $\smash{(z, \pi(i)) \in G^{[\kappa]}(T)}$ for some $\kappa$, then the weight is at least $2^{-\kappa }$ and $|\mu(\ell)_{\pi(i)}|\geq 2^{-\kappa - 1}$.
\end{proof}
 
\subsection{Case 1}\label{sec:case1}

In this section, we prove \cref{lem:main-tech} assuming that we are in the first case outlined in \cref{sec:bucketing}. As per \cref{sec:graphs-to-mean} and the first case described in \cref{sec:bucketing}, we consider a distribution $\p$ supported on $\{-1,1\}^n$, a parameter $t \in [n-1]$, and we focus on the family of directed graphs $\smash{G^{[u]}(T)}$, one for each size-$t$ subset $T$ of $[n]$.
We assume that there are parameters $\zeta', \xi$, and $d\ge 1$ such that with probability at least $\zeta'$ over the draw of $\bT \sim \calS(t)$, 
\begin{equation}\label{hehe6} 
\Prx_{\bx \sim \p_{\ol{\bT}}}\left[ d \leq \outdegree(\bx, G^{(u)}(\bT)) \leq 2d \right] \geq \xi. 
\end{equation}
Notice that we will set $\zeta' = \zeta / (2 \log n)$, so that (\ref{eq:case-1-params}) implies 
\begin{align} 
\sqrt{d} \cdot \xi \gsim \frac{\alpha}{\zeta' \log^3 n \log(1/\alpha)}. \label{eq:case-1-newvars}
\end{align}
In this case, we will show that~\eqref{eq:restrict-t} holds. The following definition and lemma will be useful:  
\begin{definition}
Let $\tau$ be a $t$-sequence of distinct indices from $[n]$ and let $y\in \{-1,1\}^n$.
We say the pair $(\tau,y)$ is \emph{\tc} if the restricted distribution $p_{|\rho(\tau,y)}$ satisfies
\begin{equation}\label{ufuf}
\left\|\mu\big(\p_{|\rho(\tau,y)}\big)\right\|_2 \geq \frac{\sqrt{d}}{20},
\end{equation} and 
  we say $(\tau,y)$ is \emph{\tpc} otherwise. 
\end{definition} 
  Note that the definition implies that a pair is either \tc or \tpc.

\begin{lemma}\label{hehe5}
Let $\pi$ be a $(t+1)$-sequence of distinct indices from $[n]$, and let $y\in \{-1,1\}^n$.
 If there are $d+1$ distinct indices $i_1,\ldots,i_{d+1} \in [t+1]$ such that 
\begin{equation}\label{haheha1}
\left(y_{\ol{S(\pi_{-i_k})}}, \pi(i_k) \right)\in G^{[u]}(S(\pi_{-i_k}))
\end{equation}
for all $k\in [d+1]$,
then $(\pi_{-i_k},y)$ is a \tc pair for all $k\in [d+1]$.
\end{lemma} 
\begin{proof}
We prove the lemma for $k=1$.
For convenience we refer to $i_1$ as $i$ and $i_2,\ldots,i_{d+1}$ as $j_1,\ldots,j_d$.
We show that for every $j=j_1,\ldots,j_d$,
$\big|\mu(p_{|\rho(\pi_{-i},y)})_{\pi(j)}\big|\ge 1/20$. The lemma then follows.

To simplify notation, we  $R=\ol{S(\pi)}\cup \{\pi(i),\pi(j)\}$. 
Let $z=y_{R}\in \{-1,1\}^{R}$. Let
\begin{align*}
a_1 &= \Prx_{\bx \sim \p}\big[ \bx_R = z \big] \qquad\qquad\hspace{.5cm} a_2 = \Prx_{\bx \sim \p}\left[ \bx_{R} = z^{(\pi(i))}\right] \\[0.5ex]
a_3 &= \Prx_{\bx \sim \p}\left[ \bx_R = z^{(\pi(j))}\right] \qquad\qquad a_4 = \Prx_{\bx \sim \p}\left[\bx_{R} = z^{(\{\pi(i),\pi(j)\})}\right] 
\end{align*}
Notice that since $(y_{\ol{S(\pi_{-i})}}, \pi(i))$ and $(y_{\ol{S(\pi_{-j})}}, \pi(j))$ are uneven edges, we have
\begin{align}
a_1 + a_3 \geq 3 \left( a_2  +a_4 \right) \qquad \text{and}\qquad a_1 + a_2 \geq 3 \left( a_3 + a_4\right). \label{eq:i-j-relate}
\end{align}
We use these two inequalities to show that
\[
\left|\mu(p_{|\rho(\pi_{-i},y)})_{\pi(j)}\right|=\frac{|a_1-a_3|}{a_1+a_3}\ge \frac{1}{20}.
\]
Suppose first that $a_3 \leq 9 a_1/10$, then
\begin{align*}
\frac{|a_1 - a_3|}{a_1 + a_3} \geq \frac{a_1/10}{a_1+a_3} \geq \frac{1}{10} \cdot \frac{1}{2} \geq \frac{1}{20}.
\end{align*}
If $a_3 \geq 9a_1/10$, then by the second equality in (\ref{eq:i-j-relate}), $a_2 \geq 17 a_1/10 + 3a_4$. 
 Furthermore, substituting into the first inequality of (\ref{eq:i-j-relate}), $a_1 + a_3 \geq 3(17a_1 / 10 + 3a_4) + 3a_4$, which implies $a_3 \geq 4 a_1$, so
\begin{align*}
 \frac{|a_1 - a_3|}{a_1 + a_3} \geq \frac{3a_3/4}{a_1 + a_3} \geq \frac{3}{4} \cdot \frac{1}{2} \geq \frac{1}{20}.
\end{align*}
This finishes the proof of the lemma.
\end{proof}  \medskip

Now we start to lower bound the expectation of $\| \mu(\p_{|\brho})\|_2$ as $\brho\sim \calD(t+1,\p)$, or equivalently as $\brho$ is drawn as $\rho(\bpi,\by)$.
For this purpose we introduce an indicator random variable $\bX_i$ for each $i\in [t+1]$:
  $\bX_i$ is $1$ if the following event holds:
\begin{equation}\label{hehe1}
 \left(\by_{\overline{S(\bpi_{-i})}}, \bpi(i) \right)
\in G^{[u]}(S(\bpi_{-i}))\ \text{and}\ \text{$(\bpi_{-i},\by)$ is \tpc} \end{equation}
On the one hand, it follows from \cref{lem:graph-to-mean} (and the first part of
  the event above) that
\[
 \big\| \mu(\p_{|\rho(\bpi,\by)})\big\|_2 \gsim \sqrt{\bX_1+\cdots
  +\bX_{t+1}}.
\]
On the other hand, it follows from \cref{hehe5} and the second part of
  (\ref{hehe1}) that $\bX_1+\cdots +\bX_{t+1}$ is at most $d$ with probability $1$.
As a result, we have 
\[
\Ex_{\bpi,\by} \left[\big\| \mu(\p_{|\rho(\bpi,\by)})\big\|_2\right]
\gsim  \Ex_{\bpi,\by}\Big[ \bX_1+\cdots
  +\bX_{t+1} \Big] \Big/{\sqrt{d}} 
\]
This simplifies the task now to bound the probability of $\bX_i=1$.
We claim that for each $i\in [t+1]$,
\begin{equation}\label{hehe3}
\Pr_{\bpi,\by}\big[\bX_i=1\big]\ge 
\Big(\zeta'\cdot \xi-\Pr_{\bpi,\by}\big[\text{$(\bpi_{-i},\by)$ is \tc}\big]\Big)\cdot \frac{d}{n}, 
\end{equation}
To see this is the case, we consider drawing $\bpi$ and $\by$
  by drawing $\by$ and $\bpi_{-i}$ first and then $\bpi(i)$.
We consider the following event $F$ over $\by$ and $\bpi_{-i}$:   
\begin{quote}Event $\bF$: $S(\bpi_{-i})$ as $\bT$ and $\by_{\overline{ S(\bpi_{-i})}}$ as $\bx$ 
  satisfy (\ref{hehe6}) 
  and $(\bpi_{-i},\by)$ is \tpc.\end{quote}
Note that it follows from our assumption at the beginning of this section
  that the first part of $\bF$ occurs with probability at least $\zeta' \cdot \xi$.
It follows that the probability of $\bF$ is at least 
\[
\zeta'\cdot \xi-\Pr_{\bpi,\by}\big[\text{$(\bpi_{-i},\by)$ is \tc}\big].
\]
Finally, conditioning on $\bpi_{-i}$ and $\by$ satisfying $\bF$, 
  $\bpi(i)$ (together with $\bpi_{-i}$ and $\by$) leads to $\bX_i=1$ if it is one of the
  at least $d$ edges of $\by_{\ol{S(\bpi_{-i})}}$ in $G^{[u]}(S(\bpi_{-i}))$. This occurs with probability at least $d/n$. 

Continuing from (\ref{hehe3}), we note that the probability of $(\bpi_{-i},\by)$ being \tc
  is the same as $(\btau,\by)$ being \tc, where $\btau$ is drawn uniformly from all
  $t$-sequences of distinct indices.
Combining everything,  
\begin{align}
\Ex_{\bpi,\by} \left[\big\| \mu(\p_{|\rho(\bpi,\by)})\big\|_2\right] \gsim \frac{1}{\sqrt{d}}\cdot  (t+1)\cdot \left( \zeta' \cdot \xi - \Prx_{\btau,\by}\big[\text{$(\btau,\by)$ is \tc}\big] \right) \cdot \frac{d}{n}. \label{eq:hehehe}
\end{align}
Thus, either the probability of $(\btau,\by)$ being \tc is at least $\zeta' \cdot \xi / 2$,
  in which case we have
\[
\Ex_{\brho\sim \calD(t,p)}\Big[\big\|\mu(p_{|\brho})\big\|_2\Big]=
\Ex_{\btau,\by} \left[\big\| \mu(\p_{|\rho(\btau,\by)})\big\|_2\right] \gsim \zeta'\cdot \xi\cdot \sqrt{d}\gsim
\frac{\alpha}{\log^3 n \log(1/\alpha)},
 \]
using \eqref{eq:case-1-newvars} for the last inequality;  or (\ref{eq:hehehe}) is lower bounded by
\[ 
\frac{1}{\sqrt{d}}\cdot (t+1)\cdot \frac{\zeta'\cdot \xi}{2}\cdot \frac{d}{n}  \gsim \frac{t }{n} \cdot \frac{\alpha}{\log^3 n \log(1/\alpha)},   \]
which shows~\eqref{eq:restrict-t} holds since $t\le n-1$.

\begin{remark}\label{remarkcase1}
Note that if $dt/n\ge 1$, one may set the right hand side of (\ref{ufuf}) to be
  $\sqrt{dt/n}/20$ instead of $\sqrt{d}/20$, and this would result in a stronger lower bound 
  in (\ref{eq:restrict-t}) where we can replace the $t/n$ by $\sqrt{t/n}$.
However, we do not have any control on the parameter $d$, which could be $1$ in the worst case. %
\end{remark}
 
\subsection{Case 2}\label{sec:case2}

In this section, we prove \cref{lem:main-tech} assuming that we are in the second case outlined in \cref{sec:bucketing}. As per \cref{sec:graphs-to-mean} and the second case of \cref{sec:bucketing}, we consider a distribution $\p$ over $\smash{\bool^{n}}$, a parameter $t \in [n-1]$, and a parameter $\smash{\kappa \in [O(\log(n/\alpha))]}$; we focus on the family of directed graphs 
$G^{[\kappa]}(T)$ for each $t$-subset $T$ of $[n]$.
 
From Case 2 of \cref{sec:bucketing}, we assume there are parameters $\zeta', \xi$, and $d$ such that with probability at least $\zeta'$ over the draw of $\bT \sim \calS(t)$, we have
\begin{align}
\Prx_{\bx \sim \p_{\ol{\bT}}}\left[ d \leq \outdegree\big(\bx, G^{[\kappa]}(\bT)\big) \leq 2d\right] \geq \xi, \label{eq:degree-prob-case2}
\end{align}
where we set $\zeta' = \zeta / (2 \log n \cdot O(\log(n/\alpha)))$. Thus, (\ref{eq:case-2-params}) implies
\[ \sqrt{d} \cdot \xi \gsim \frac{\alpha \cdot 2^{\kappa}}{\zeta' \log^3 n \log^2(n/\alpha) \log(1/\alpha)}.  \]
The goal is to lower bound 
  the expectation of $\| \mu(\p_{|\brho})\|_2$ with $\brho\sim \calD(t,p)$ or $\brho\sim \calD(t+1,p)$. For this, we consider a similar notion of bad pairs to that of \cref{sec:case1}, and start with a lemma that will be useful to the analysis. 
\begin{definition}
Let $\gamma\ge 1$.\footnote{We will set $\gamma=d$ at the end but would like to keep it as a 
  parameter for the discussion in \cref{remarkcase2}.}{} We say a restriction $\rho$ with $t$ stars is \emph{\gtc} if
\[ \left\| \mu\big(\p_{|\rho }\big) \right\|_2 \geq \frac{\sqrt{\gamma}}{2^{\kappa+1}}, \]
and we say that $\rho$ is \emph{\gtpc} otherwise. 
\end{definition}
Similar to before, a restriction is either \gtc or \gtpc.
\begin{lemma}\label{capcapcap}
Let $\pi$ be a $(t+1)$-sequence of distinct indices from $[n]$ and let $y \in \{-1,1\}^n$. Suppose that $i\in[t+1]$ is such that both 
  restrictions $\rho(\pi_{-i}, y)$ and $\smash{\rho(\pi_{-i}, y^{(\pi(i))})}$ are \gtpc. Then we have
\[ \sum_{j \in [t+1] \setminus \{i\}} \left(\mu\big(\p_{|\rho(\pi, y)}\big)_{\pi(j)}\right)^2 < \frac{\gamma}{2^{2\kappa+2}}. \]
\end{lemma}

\begin{proof}
Let $a$ be the following conditional probability:
\[ a = \Prx_{\bx \sim \p}\left[ \bx_{\pi(i)} = y_{\pi(i)}\hspace{0.06cm} \big|\hspace{0.06cm} \bx_{\ol{S(\pi)}} = y_{\ol{S(\pi)}}\hspace{0.04cm}\right],\]
and notice that for all $j \in [t+1] \setminus \{i\}$, we have
\[ \mu\big(\p_{\rho(\pi, y)}\big)_{\pi(j)} = a \cdot \mu\big(\p_{|\rho(\pi_{-i}, y)}\big)_{\pi(j)} + (1-a)\cdot \mu\big(\p_{|\rho(\pi_{-i},y^{(\pi(i))})}\big)_{\pi(j)}.  \] 
By Jensen's inequality,
\begin{align*}
\sum_{j \in [t+1] \setminus \{i\}} \left(\mu\big(\p_{|\rho(\pi, y)}\big)_{\pi(j)}\right)^2 
 \leq \sum_{j \in [t+1] \setminus \{i\}} \left( a \cdot \left(\mu\big(\p_{|\rho(\pi_{-i}, y)}\big)_{\pi(j)}\right)^2 + (1-a) \cdot \left(\mu\big(\p_{|\rho(\pi_{-i}, y^{(\pi(i))})}\big)_{\pi(j)}\right)^2  \right)
\end{align*}
which less than $\gamma/2^{2\kappa +2}$
since $\rho(\pi_{-i}, y)$ and $\rho(\pi_{-i}, y^{(\pi(i))})$ are both \gtpc.
\end{proof}
 
We proceed in a similar fashion to Case 1. 
To bound the expectation of $\| \mu(\p_{|\brho})\|_2$ as $\brho\sim \calD(t+1,\p)$, or equivalently as $\brho$ is drawn as $\rho(\bpi,\by)$, we introduce an indicator random variable $\bX_i$ for each $i\in [t+1]$:
  $\bX_i$ is $1$ if the following event holds:
\begin{equation}\label{eq:eventevent}
\left(\by_{\ol{S(\bpi_{-i})}}, \bpi(i)\right) \in G^{[\kappa]}(S(\bpi_{-i})) \text{ and both } \rho(\bpi_{-i}, \by)\ \text{and}\ \rho(\bpi_{-i}, \by^{(\bpi(i))})  \text{ are \gtpc}
\end{equation} 
First, it follows from \cref{lem:graph-to-mean} (and the first part of
  the event above) that
\[
 \left\| \mu\big(\p_{|\rho(\bpi,\by)}\big)\right\|_2 \ge \frac{1}{2^{\kappa+1}}\cdot \sqrt{\bX_1+\cdots
  +\bX_{t+1}}.
\]
But this, along with \cref{capcapcap}
 and the second part of
  (\ref{eq:eventevent}), implies that $\bX_1+\cdots +\bX_{t+1}$ is at most $\gamma$ with probability $1$.
As a result, we have 
\[
\Ex_{\bpi,\by} \bigg[\left\| \mu\big(\p_{|\rho(\bpi,\by)}\big)\right\|_2\bigg]
\gsim  \frac{1}{2^{\kappa+1}}\cdot \Ex_{\bpi,\by}\Big[ \bX_1+\cdots
  +\bX_{t+1} \Big]\cdot \frac{1}{\sqrt{\gamma}} 
\] 

\noindent This reduces the task now to bounding the probability of $\bX_i=1$.

We start with some notation. 
Let $T$ be a $t$-subset of $[n]$ and let $z\in \{-1,1\}^{\ol{T}}$.
We use $\rho(z)$ to denote the restriction $\rho\in \{-1,1,*\}^n$ with $\rho_i=z_i$ for all $\smash{i\in \ol{T}}$
  and $\rho_i=*$ for all $i\in T$.
To lower bound the probability of $\smash{\bX_i=1}$, we define the following two disjoint 
  subsets of $\smash{\{-1,1\}^{\ol{T}}}$ for each $t$-subset $T$:
\begin{align*}
 A_{T} &= \left\{ z \in \{-1,1\}^{\ol{T}} : d\le \outdegree(z,G^{[\kappa]}(T))\le 2d\ \text{and $\rho(z)$ is \gtpc}\right\} \\[0.8ex]
 B_T &= \left\{ w \in \{-1, 1\}^{\ol{T}} : \text{$\rho(w)$ is \gtc} \right\}. 
 \end{align*}
Then the probability of $\bX_i=1$  (i.e., the event described in (\ref{eq:eventevent}))
  is at least the probability of the following event $\mathbf{E}$,  where 
  we first draw a $t$-subset $\bT$, then $\bz$ from $p_{\ol{\bT}}$, and 
  finally draw $\bi$ from $\ol{\bT}$:
\begin{align*}
\text{Event $\mathbf{E}$: 
  $(\bz,\bi)\in G^{[\kappa]}(\bT)$, $\bz\in A_{\bT}$ and
  $\bz^{(\bi)}\notin  B_{\bT}$}
\end{align*}
We write the probability of $\mathbf{E}$ as
\begin{align}
\Prx_{\bT,\bz,\bi}\big[\mathbf{E}\big]&=\Prx_{\bT,\bz,\bi}\Big[(\bz,\bi)\in G^{[\kappa]}(\bT)\wedge \bz\in A_\bT\Big]\nonumber\\[0.5ex]
  &\hspace{1cm}-\Prx_{\bT,\bz,\bi}\Big[(\bz,\bi)\in G^{[\kappa]}(\bT) \wedge \bz\in A_\bT \wedge 
  \bz^{(\bi)}\in B_\bT\Big].\label{haheuu}
\end{align}
The first probability on the right hand side of (\ref{haheuu})
  is at least 
\[\left( \zeta' \cdot \xi - \Prx_{\brho\sim \calD(t,p)}\big[\brho\ \text{is \gtc}\hspace{0.03cm} \big]\right)\cdot \frac{d}{n-t}.\]
To see this we first draw $\bT$ and $\bz$ and impose the condition 
  that $\bz\in A_{\bT}$.
Similar to the analysis in Case 1, this happens with probability at least
\[
\zeta'\cdot \xi-\Prx_{\bT,\bz}\big[\bz\in B_\bT\hspace{0.02cm}\big]
=\zeta'\cdot \xi-\Prx_{\brho\sim \calD(t,p)}\big[\brho \text{\ is \gtc}\hspace{0.03cm}\big].
\]
Then we draw $\bi$ and the probability we get an outgoing edge in $G^{[\kappa]}(\bT)$ is at least $d/(n-t)$.

The probability we subtract on the right hand side of (\ref{haheuu}) can be written as
\begin{align*}
&\frac{1}{\binom{n}{t}} \sum_{T \in \calP(t)} \sum_{z \in A_{T}} \Prx_{\bx \sim\p_{\ol{T}}}\big[\bx=z\big] \cdot \sum_{i\in \ol{T}} \frac{1}{|\ol{T}|} \cdot \ind\left\{ (z, i) \in G^{[\kappa]}(T) \wedge z^{(i)} \in B_T\right\} \\[0.8ex]
&\qquad \leq \frac{3}{\binom{n}{t}}\cdot \frac{1}{n-t} \sum_{T \in \calP(t)} \sum_{z \in A_T} \sum_{i \in \ol{T}} \hspace{0.15cm}\Prx_{\bx \sim \p_{\ol{T}}}\big[ \bx = z^{(i)}\big] \cdot \ind\left\{ (z, i) \in G^{[\kappa]}(T) \wedge z^{(i)} \in B_T \right\} \\[0.8ex]
&\qquad = \frac{3}{\binom{n}{t}}\cdot \frac{1}{n-t} \sum_{T \in \calP(t)} \sum_{w \in B_T} \sum_{i\in \ol{T}}
\hspace{0.1cm}\Prx_{\bx\sim p_{\ol{T}}} \big[\bx=w\big]\cdot \ind\left\{(w^{(i)},i)\in G^{[\kappa]}(T)\wedge w^{(i)}\in A_T\right\}\\[0.8ex]
&\qquad = \frac{3}{\binom{n}{t}}\cdot \frac{1}{n-t} \sum_{T \in \calP(t)} \sum_{w \in B_T}  
 \Prx_{\bx\sim p_{\ol{T}}} \big[\bx=w\big]\cdot \left[\hspace{0.04cm}\text{number of edges from $A_T$ to $w$ in $G^{[\kappa]}$}\hspace{0.04cm}\right].
\end{align*}
where the first inequality used 
  $\Prx_{\bx}[\bx=z]\le 3\cdot \Prx_{\bx}[\bx=z^{(i)}]$ because each $\{z,z^{(i)}\}$ is an even edge,
and the second equation is just a change of variable.
Given that every vertex in $A_T$ has out-degree at most $2d$ in $G^{[\kappa]}$,
  we can now apply \cref{lem:greedy} to conclude that 
  the number of edges from $A_T$ to every $w\in B_T$ is at most $2d$.
As a result, the probability we subtract can be bounded from above by
\[
\frac{3}{\binom{n}{t}}\cdot \frac{1}{n-t} \sum_{T \in \calP(t)} \sum_{w \in B_T} \Prx_{\bx \sim \p_{\ol{T}}}\left[ \bx = w \right] \cdot 2d \leq  \frac{6d}{n-t}\cdot \Prx_{\bT,\bz}\big[\bz\in B_\bT\hspace{0.02cm}\big]. 
\]

Therefore, we have (\ref{haheuu}) that
\begin{align*}
\Prx_{\bpi,\by}\big[\bX_i=1\big]\ge 
\Prx_{\bT, \bz, \bk}\big[\mathbf{E}\big] 
\ge \left(\zeta' \cdot \xi - 7\cdot  \Prx_{\brho \sim \calD(t, \p)} \big[ \text{$\brho$ is \gtc} \big] \right)\cdot  \frac{d}{n-t}. 
\end{align*}

To conclude the proof of \cref{lem:main-tech} for Case 2, we set $\gamma=d$.
Then 
  either we have 
  \[
  \Prx_{\brho \sim \calD(t, \p)}\big[ \text{$\brho$ is \gtc} \big] \geq 
    \frac{\zeta' \cdot \xi}{ 14},
  \]
  which implies
\[ \Ex_{\brho  \sim \calD(t, \p)}\bigg[\left\|\mu\big(\p_{|\brho'}\big)\right\|_2 \bigg] \geq \frac{\zeta' \cdot \xi}{14} \cdot  \frac{\sqrt{d}}{2^{\kappa+1}} \gsim \frac{\alpha  }{ \log^3 n \log^2(n/\alpha) \log(1/\alpha)},\] 
or we have  \begin{align*}
\Ex_{\brho \sim \calD(t+1, \p)}\bigg[ \left\| \mu\big(\p_{|\brho}\big)\right\|_2 \bigg] \gsim 
\frac{t+1}{2^{\kappa}\cdot \sqrt{d}}\cdot \frac{d}{n-t}\cdot \frac{\zeta' \cdot \xi}{14}
\gsim \frac{t}{n}\cdot \frac{\alpha}{\log^3 n \log^2(n/\alpha)\log (1/\alpha)}.  
\end{align*}
Taking the sum always gives us (\ref{eq:restrict-t}).\medskip

\begin{remark}\label{remarkcase2}
Similar to Remark \ref{remarkcase1}, note that if $dt/n\ge 1$, one may 
  set $\gamma= dt/n$ instead of $d$, and this would result in a stronger lower bound 
  in (\ref{eq:restrict-t}) where we can replace the $t/n$ by $\sqrt{t/n}$.
Similarly to Remark~\ref{remarkcase1}, we do not have control over the parameter $d$, which could be $1$ in the worst case.%
\end{remark}

\section{Mean Testing}
For any $\dimss \in \N$ and any distribution $\p$ supported on $\{-1,1\}^{\dimss}$, we consider $\bX=(\bx^{(1)},\dots,\bx^{(\ns)}),\bY=(\by^{(1)},\dots,\by^{(\ns)})$, a set of $2\ns$ i.i.d.\ samples from $\p$, and let
\[
    \bar{\bX} \eqdef \frac{1}{\ns}\sum_{i=1}^\ns \bx^{(i)}\,,\quad \bar{\bY} \eqdef \frac{1}{\ns}\sum_{i=1}^\ns \by^{(i)}
\]
be the empirical means in $\R^\dimss$. Our core test statistic will take $2\ns$ i.i.d.\ samples from $\p$, and compute the expression
\begin{equation}\label{eq:statistic:z}
  \bZ \eqdef \dotprod{\bar{\bX}}{ \bar{\bY}}.
\end{equation}
We will write $\mu(\p) = \Ex_{\bx \sim \p}[\bx] \in [-1,1]^{\dimss}$ as the mean vector and $\Sigma(\p)\in \R^{\dimss \times \dimss}$ as the symmetric matrix with $\Sigma(\p)_{ij} = \Ex_{\bx \sim \p}[\bx_i \bx_j]$ for all $i, j \in [\dimss]$.\footnote{In particular, note that due to the diagonal terms we have $\norm{\Sigma(\p)}_F \geq \sqrt{\dimss}$ for all $\p$.}

\begin{lemma}  \label{lemma:variance:z:general:p} The random variable $\bZ$ obtained from two tuples of $\ns$ samples from $\p$ satisfies
\begin{align}
  \Ex[\bZ] &= \dotprod{\Ex[\bar{\bX}]}{\Ex[\bar{\bY}]} = \dotprod{\mu(\p)}{\mu(\p)} = \normtwo{\mu(\p)}^2\, , \label{eq:expectation:z} \\
  \Var[\bZ] &\leq \frac{1}{\ns^2}\norm{\Sigma(\p)}_F^2 + \frac{4}{\ns}\normtwo{\mu(\p)}^2\norm{\Sigma(\p)}_F. \label{eq:variance-z}
\end{align}
\end{lemma}

The proof of (\ref{eq:expectation:z}) follows from the fact that $\bar{\bX}$ and $\bar{\bY}$ are independent, and (\ref{eq:variance-z}) is a computation which appears in Appendix~\ref{sec:mean-appendix}. We define the \emph{blowup distribution} $\blowup{\p}$ as the distribution on $\bool^{\dimss^2}$ such that, ordering $[\dimss]\times[\dimss]$ in the lexicographic way, $\blowup{\p}$ is the distribution of the vector
\begin{align}
    (\bx_i \bx_j)_{(i,j)\in[\dimss]\times[\dimss]} = (\bx_1\bx_1, \bx_1\bx_2,\dots,\bx_1\bx_\dimss,\bx_2\bx_1,\bx_2\bx_\dimss,\dots,\bx_\dimss \bx_\dimss) \label{eq:convert-sample}
\end{align}
when $\bx\sim\p$. For $k\in\N$, we let $\blowupit{k}{\p}$ denote the distribution over $\bool^{\dimss^{2^k}}$ obtained by iterating this process $k$ times, so that in particular $\blowupit{0}{\p}=\p$. Note that, for any $k\geq 0$, a sample from $\blowupit{k}{\p}$ can be obtained from a sample from $\p$ in time $O(\dimss^{2^k})$. 

\begin{fact}\label{fact:mean:blowup}
  For any distribution $\p$ over $\bool^\dimss$ and $k\in\N$, we have
  $
        \normtwo{ \mu( \blowupit{k+1}{\p} )}^2 = \norm{\Sigma(\blowupit{k}{\p})}_F^2
  $.
\end{fact}

Consider the following threshold test:
\begin{algorithm}[H]
  \begin{algorithmic}[1]
    \Require A threshold $\tau>0$ and a (multi)set $S = \{ x_1, \dots, x_{\ns}, y_1, \dots, y_{\ns}\}$.
    \State Compute $Z$ from the samples in $S$ according to (\ref{eq:statistic:z}).
    \If{ $\smash{Z > \tau}$ } \Return \reject
    \EndIf
    \State \Return \accept 
  \end{algorithmic}
  \caption{$\textsc{ThresholdZTest}(\tau, S)$}
\end{algorithm}
\noindent Hereafter, for $\epsilon\in (0,1]$, we consider the sequence $(\tau_k)_{k \in \Z^{\geq 0}}$ of real numbers, defined recursively as 
\begin{equation}\label{eq:setting:tauk}
      \tau_k \eqdef
      \begin{cases}
          \frac{\dst^2n}{2} &\text{ if }k = 0\\
          \frac{1}{5000} \cdot \ns^2 \tau_{k-1}^2 &\text{ if }k\geq 1
      \end{cases}
\end{equation}
In particular, writing $a\eqdef 1/5000$, we have 
\begin{align}
\tau_k &= \frac{1}{a\ns^2}\cdot \left(\frac{a\ns^2\dst^2n}{2}\right)^{2^k} \label{eq:unrolled-tau}
\end{align}
for $k\geq 0$. We are now ready to state the main testing algorithm. The algorithm takes as input a distribution $\p$ which is supported on $\{-1,1\}^\dims$, and is written with two unspecified parameters, $k_0$ and $q$. The parameter $k_0$ denotes the number of rounds and $q$ denotes the sample complexity.

\begin{algorithm}[H]
  \begin{algorithmic}[1]
    \Require A distribution $\p$ supported on $\{-1,1\}^\dims$.
    \State Draw a set $\bS$ of $2\ns$ i.i.d.\ random samples from $\p$. 
    \ForAll{$0\leq k\leq k_0$}
      \State Set $\tau_k$ as in~\eqref{eq:setting:tauk}.
      \State Convert the $2\ns$ samples from $\bS$ to a (multi)set $\bS^{(k)}$ of samples from $\blowupit{k}{\p}$ as in (\ref{eq:convert-sample}).
        \If{ $\textsc{ThresholdZTest}(\tau_k,\bS^{(k)})$ returns \reject } \label{step:call:threshold}
            \Return \reject \label{step:majority:reject} 
        \EndIf
    \EndFor
    \State \Return \accept  \Comment{All $k_0 + 1$ tests were successful}
  \end{algorithmic}
  \caption{$\textsc{MeanTester}(\p, \dst)$}\label{algo:testing:l2:unrolled}
\end{algorithm}

\begin{theorem}
  \label{theo:testing:l2:unrolled}
 Fix any $k_0 \in \N$. There exists an algorithm (\cref{algo:testing:l2:unrolled}) which, given sample access to an arbitrary distribution $\p$ on $\bool^\dims$ and a parameter $\dst \in( 0,1]$, has the following behavior:
  \begin{itemize}
    \item If $\p$ is the uniform distribution, the algorithm outputs \accept with probability at least $2/3$;
    \item If $\p$ satisfies $\normtwo{\mu(\p)} \geq \dst\sqrt{n}$,
  the algorithm outputs \reject with probability at least $2/3$.
  \end{itemize}
  These guarantees hold as long as
  \[
    \ns \gsim \max\left\{ \frac{{1}}{\dst^2\sqrt{n}}, \left(\frac{1}{\dst^2}\right)^{\frac{2^{k_0+1}}{2^{k_0+2} - 2}}\right\}\,.
  \] 
 The algorithm runs in time $O\left(q \cdot \dims^{2^{k_0}} \right)$.
\end{theorem}

In particular, by setting $k_0 = \log \log \dims$, we obtain an algorithm for distinguishing the uniform distribution from a distribution $\p$ on $\{-1,1\}^{\dims}$ with $\|\mu(\p)\|_2 \geq \dst\sqrt{n}$ which runs in time $n^{\Theta(\log n)}$ and has sample complexity
\[ O\left( \max\left\{\frac{1}{\dst ^2\sqrt{n}}, \frac{1}{\dst} \right\}\right).\] 

\begin{remark}
  \label{rk:improved:runtime}
  As stated the algorithm is not computationally efficient, as for $k_0 = \log \log \dims$ it runs in superpolynomial time $\dims^{O(\log\log\dims)}$. This follows from using the obvious but naive approach to computing the statistic $Z$ in~\eqref{eq:statistic:z} for the various blowup distributions $\blowupit{k}{\p}$; however, this can be greatly improved by computing this statistic in a more careful way, rephrasing it as a sum of inner products of tensor products of the original samples and relying on the mixed-product property of tensor products. Doing so results in a running time polynomial in both $q$ and $\dims$; we refer the reader to the proof of~\cite[Theorem~6]{chen2020learning} for details.
\end{remark}

\begin{proofof}{\cref{theo:testing:l2:unrolled}}
The proof will proceed as follows: we first show that, when $\p$ is the uniform distribution $\uniform$ (the completeness case), then all $k_0 +1$ tests, when run on Line~\ref{step:call:threshold} of \cref{algo:testing:l2:unrolled}, return \accept with high probability. To do so, notice that~\cref{step:call:threshold} of \cref{algo:testing:l2:unrolled} considers samples from $\blowupit{k}{\uniform}$. Hence, we analyze the mean and variance of the statistic for each $\blowupit{k}{\uniform}$, and apply Chebyshev's inequality to show that, for any given $k$, each call to~\cref{step:call:threshold} then returns \accept with probability at least $1-2^{-k}/6$. By a union bound over all $k$, we get that overall all calls will return \accept with probability at least $1-\sum_{k=0}^\infty 2^{-k}/6 = 2/3$.
\begin{lemma}\label{lemma:mean:blowup:uniform}
  For the uniform distribution $\uniform$ over $\bool^\dims$ and $k\in\N$, we have
  \[
        \norm{\Sigma(\blowupit{k}{\uniform})}_F^2 \leq (\dims 2^k)^{2^{k}}\,.
  \]
\end{lemma}
\begin{proof}
  Let $K\eqdef 2^{k+1}$. Since whenever $\bx_i^2=1$ for all $i\in[\dims]$ when $\bx \sim \Uniform$, and the $\bx_i$'s are independent and have zero-mean, 
 we have that, for any ordered $K$-tuple $(i_1,\dots,i_K)\in[\dims]^K$
  \[
    \Ex_{x\sim\uniform}\left[\prod_{\ell=1}^K \bx_{i_\ell} \right] = 
    \begin{cases}
        1 & \text{if each $i_\ell$ appears an even number of times }\\
        0 & \text{otherwise.}
    \end{cases}
  \]
  Therefore, $\norm{\Sigma(\blowupit{k}{\uniform})}_F^2$ is upper bounded by the number of ordered $K$-tuples of $[\dims]$ in which each index appears an even number of times. This in turn is at most 
  \[
      \dims(K-1)\cdot \dims(K-3) \cdot \dots \cdot \dims = \left(\frac{\dims}{2}\right)^{K/2}\cdot \frac{K!}{(K/2)!}
      \leq \left(\frac{\dims}{2}\right)^{K/2} \cdot K^{K/2}\,,
  \]
which one can see as follows: from $K$ variables, we may choose the value of the first (there are $\dims$ choices), and then pair this variables with an other one to which we assign the same value (there are $K-1$ such choices). We then recurse on the remaining $K-2$ variables: this process, though it may lead to double-counting, ensures that each value appears an even number of times in the resulting $K$-tuple~--~as we always assign any chosen value to two variables.
\end{proof}

  \noindent Combining~\cref{fact:mean:blowup} and~\cref{lemma:mean:blowup:uniform}, this implies
  \begin{align}
        \normtwo{ \mu( \blowupit{k}{\uniform} )}^2 \leq(\dims 2^{k-1})^{2^{k-1}}
        =  \sqrt{(\dims 2^k / 2)^{2^{k}}}. \label{eq:mean:blowup:uniform}
  \end{align}

\begin{lemma}[Completeness]\label{lemma:completeness}
There exists a large enough universal constant $C> 0$, such that for any $k\in\N$ where $2^k \leq \log^2 n$. If $\ns \geq C/\dst^2\sqrt{\dims}$, then letting $\bS = \{ \bx_1, \dots, \bx_{\ns}, \by_1, \dots, \by_{\ns}\}$ be $2\ns$ i.i.d.\ samples from $\blowupit{k}{\uniform}$. Then,
  \begin{align}
  \Prx_{\bS}\left[ \textsc{ThresholdZTest}(\tau_k, \bS) \text{ outputs } \reject \right] &\leq \frac{2^{-k}}{6}. \label{eq:prob-lb} 
  \end{align}
\end{lemma}
\begin{proof}
Given $\bS$, let $\bZ$ be the random variable given by (\ref{eq:statistic:z}), so $\textsc{ThresholdZTest}(\tau_k, \bS)$ rejects whenever $\bZ > \tau_k$. From~\cref{lemma:variance:z:general:p}, (\ref{eq:unrolled-tau}) and (\ref{eq:mean:blowup:uniform}), for all $k \in \N$, $\Ex[\bZ] = \|\mu(\blowupit{k}{\uniform})\|_2^2$ satisfies 
\begin{align}
\|\mu(\blowupit{k}{\uniform})\|_2^2 \leq \frac{\tau_k}{6\cdot 2^k}. \label{eq:relation}
\end{align}
The case of $k=0$ is trivial, and the case $k=1$ follows from our choice of $q>C/\epsilon^2\sqrt{n}$.
As a result, the left-hand side of (\ref{eq:prob-lb}) is at most
\begin{align}
\Prx_{\bS}\left[ \bZ > \tau_k \right] &\leq \Prx_{\bS}\left[ |\bZ - \Ex[\bZ]| > \frac{\tau_k}{2}\right] \leq \dfrac{4 \Var[\bZ]}{ \tau_k^2 } \label{eq:eeee} \\
	&\leq \frac{4}{\tau_k^2} \left( \frac{1}{q^2} \|\Sigma(\blowupit{k}{\uniform}\|_F^2 + \frac{4}{q} \|\mu(\blowupit{k}{\uniform})\|_2^2 \|\Sigma(\blowupit{k}{\uniform})\|_2 \right) \label{eq:eheh} \\
	&\leq \frac{4}{\tau_k^2 q^2} \cdot \frac{\tau_{k+1}}{40\cdot 2^k} + \frac{16}{\tau_k^2 q} \cdot \frac{\tau_k}{20\cdot 2^k} \cdot \sqrt{\frac{\tau_{k+1}}{10}} < \frac{1}{6\cdot 2^k}, \label{eq:final}
\end{align}
where we used Chebyshev's inequality in (\ref{eq:eeee}),~\cref{lemma:variance:z:general:p} in (\ref{eq:eheh}), and~\cref{fact:mean:blowup} and $\tau_{k+1} = a\tau_k^2 q^2$, as well as (\ref{eq:relation}) in (\ref{eq:final}).
\end{proof}

\medskip
\noindent By a union bound over all $k$, we thus get that the algorithm, when run on the uniform distribution $\uniform$, outputs $\reject$ with probability at most $\sum_{k=0}^\infty \frac{2^{-k}}{6} = 1/3$. For the soundness case, the following lemma will be useful. 
\begin{lemma}\label{lemma:recurse:blowup:mean}
Let $\p$ be a distribution supported on $\{-1,1\}^{\dimss}$, satisfying
\begin{enumerate}
\item $\|\mu(\p)\|_2^2 > 2\tau$, and
\item When $\bS = \{ \bx_1, \dots, \bx_q, \by_1, \dots, \by_q\}$ is set to $2q$ i.i.d.\ samples from $\p$, 
\[ \Prx_{\bS}\left[ \textsc{ThresholdZTest}(\tau, \bS) \text{ outputs \accept}\right] \geq \frac{1}{3}. \]
\end{enumerate}  
Then $\normtwo{ \mu( \blowup{\p} )}^2 \geq \frac{1}{48^2} \cdot \tau^2\ns^2$.
\end{lemma}
\begin{proof}
Recall that $\textsc{ThresholdZTest}(\tau, \bS)$ outputs accept if $\bZ \leq \tau$, where $\bZ$ is set according to $\bS$ by (\ref{eq:statistic:z}).~\cref{lemma:variance:z:general:p} implies $\Ex[\bZ] = \|\mu(p)\|_2^2 > 2\tau$, so
\begin{align}
\frac{1}{3} &\leq \Prx_{\bZ}\left[ \bZ \leq \tau \right] \leq \Prx_{\bZ}\left[ |\bZ - \Ex[\bZ]| \geq \frac{\Ex[\bZ]}{2}\right] \leq \frac{4}{\|\mu(\p)\|_2^4} \left(\frac{1}{q^2} \norm{ \Sigma(\p) }_F^2 + \frac{4}{q} \|\mu(\p)\|_2^2 \norm{ \Sigma(\p) }_F \right), \label{eq:asdf}
\end{align}
where we used Chebyshev's inequality for the last inequality. Hence, writing $\|\mu(\blowup{\p})\|_2^2$ for $\norm{ \Sigma(\p) }_F^2$ by~\cref{fact:mean:blowup}. In particular, at least one of the two terms in the right-most side of (\ref{eq:asdf}) is at least $1/6$, and thus either $\normtwo{ \mu( \blowup{\p} )}^2 \geq \tau^2\ns^2/6$ or $\normtwo{ \mu( \blowup{\p} )} \geq \tau\ns/48$.
\end{proof}

\begin{lemma}[Soundness]\label{lemma:soundness}
There exists a large enough $C > 0$ such that setting 
\[ q \geq \left(\frac{C}{\eps^2} \right)^{\frac{2^{k_0 + 1}}{2^{k_0+2}-2}}, \] 
the following holds. For any distribution $\p$ on $\bool^{\dims}$ with $\normtwo{\mu(\p)} > \dst\sqrt{n}$, there is some $k \in \{0, \dots, k_0\}$ such that letting $\bS = \{ \bx_1, \dots, \bx_q, \by_1, \dots, \by_q\}$ be $2q$ i.i.d.\ samples from $\blowupit{k}{\p}$,
\[ \Prx_{\bS}\left[ \textsc{ThresholdZTest}(\tau_k, \bS) \text{ outputs } \reject\right] \geq \frac{2}{3}. \]
\end{lemma}

\begin{proof}
Assume for the sake of contradiction that for every $k \in \{0, \dots, k_0\}$, drawing $\bS$ from $\blowupit{k}{\p}$ satisfies $\textsc{ThresholdZTest}(\tau_k, \bS)$ outputs \accept with probability at least $1/3$. Then, by~\cref{lemma:recurse:blowup:mean} and a simple induction using the definition of $\tau_k$ in (\ref{eq:setting:tauk}) and the fact that $a < 1/(48^2 \cdot 2)$, every $k \in \{ 0, \dots, k_0 + 1\}$ satisfies $\|\mu(\blowupit{k}{\p})\|_2^2 \geq 2 \tau_{k}$. Furthermore, we always have $\|\mu(\blowupit{k}{\p})\|_2^2 \leq n^{2^k}$, so by (\ref{eq:unrolled-tau}), we may {lower bound $\normtwo{ \mu(\blowupit{k_0+1}{\p}) }^2$ as}
\begin{align*}
\frac{2}{a q^2} \left( \frac{a q^2 \eps^2 n}{2}\right)^{2^{k_0+1}} &\geq  q^{2^{k_0+2} - 2} (\dst\sqrt{n})^{2^{k_0+2}} \left(\frac{a}{2}\right)^{2^{k_0+1}-1},
\end{align*}
which contradicts the upper bound of $n^{2^{k_0+1}}$ for the setting of $q$ when $C > 0$ is a large enough constant.
\end{proof}

\medskip
As per the foregoing discussion and \cref{lemma:soundness}, there exists a parameter $0\leq k\leq k_0$ such~that $\textsc{ThresholdZTest}(\tau_{k},\bS)$ returns \reject with probability at least $2/3$ when $\bS$ is drawn from $\blowupit{k}{\p}$. For this setting of $k$, the algorithm will return \reject in~\cref{step:majority:reject} with probability at least $2/3$. Finally, from the above analysis, the sample complexity $q$ is set high enough to satisfy the constraints of \cref{lemma:completeness} and \cref{lemma:soundness}. 
\end{proofof}

\subsection{Application to Gaussian Mean Testing}\label{ssec:testing:mean:gaussian}

\begin{theorem}\label{theo:testing:l2:gaussian}
  There exists an algorithm which, given $\ns$ i.i.d.\ samples from an arbitrary Gaussian distribution $\p$ on $\R^\dims$ and a distance parameter $\dst \in (0,1]$, has the following behavior:\vspace{0.03cm}\begin{flushleft}
  \begin{itemize}
    \item If $\p$ is the standard Gaussian $\gaussian{0_{\dims}}{\mathrm{I}_{\dims}}$, then it outputs \accept with probability at least $2/3$;\vspace{0.05cm}
    \item If $\p$ is some $\gaussian{\mu}{\Sigma}$ with $\normtwo{\mu} > \dst$ (and any $\Sigma$), then it outputs \reject with probability at least~$2/3$.\vspace{0.03cm}
  \end{itemize}\end{flushleft}
  These guarantees hold as long as 
  \[
    \ns \geq C\cdot \frac{\sqrt{\dims}}{\dst^2}\,,
  \] where $C>0$ is an absolute constant, and the algorithm runs in time $\poly(\ns,\dims^{\log\dims})$. Moreover, any algorithm for this task must have sample complexity $\bigOmega{\dims^{1/2}/\dst^2}$.
\end{theorem}
\begin{proof}
The idea is to reduce the above question to $\lp[2]$ mean testing over $\bool^\dims$, and invoke~\cref{theo:testing:l2:unrolled}. The natural approach is, given a sample $x\in\R^\dims$ from the unknown Gaussian $\p$, to convert it to $y\in\bool^\dims$ by setting $y_i \eqdef \sign(x_i)$ for all $i\in[\dims]$. This idea, used e.g., in~\cite{CanonneKMUZ19} in the case of identity-covariance Gaussian distributions, clearly maps the standard Gaussian $\gaussian{0_{\dims}}{\mathrm{I}_{\dims}}$ to the uniform distribution $\uniform$ on $\bool^\dims$; therefore, the crux is to argue about the soundness case, when $\normtwo{\mu(\p)} > \dst$: we will obtain a distribution $\q$ on $\bool^\dims$ with arbitrary covariance matrix, and need to show that $\normtwo{\mu(\q)} \gtrsim \dst$.

Let $\q$ denote the distribution on $\bool^\dims$ obtained in the above fashion from $\p=\gaussian{\mu}{\Sigma}$, and note that the linear-time transformation maps a sample from $\p$ to a sample from $\q$. Further, it is easy to check that
\[
    \normtwo{\mu(\q)}^2 = \sum_{i=1}^\dims \bE{y\sim\q}{ y_i }^2 = \sum_{i=1}^\dims (2\probaDistrOf{x\sim\p}{x_i > 0}-1)^2
    = \sum_{i=1}^\dims \Paren{\operatorname{Erf}\!\Paren{\frac{\mu_i}{\sqrt{2}\Sigma_{ii}}}}^2
\]
and, from a relatively straightforward analysis of the error function $\operatorname{Erf}$, we get that $\operatorname{Erf}(x)^2 \geq \operatorname{Erf}(1)^2\cdot \min(x^2,1) > \frac{2}{3}\min(x^2,1)$ for all $x\in\R$. Therefore, we get, whenever $\normtwo{\mu(\p)} > \dst$, that
\[
    \normtwo{\mu(\q)}^2 > \frac{1}{\max_{1\leq i\leq \dims} \Sigma_{ii}^2 } \cdot \frac{\dst^2}{3}\,.
\]
Thus, it suffices to check (with high probability, say $11/12$) that (i)~all $\Sigma_{ii} = \Var_{\p}[x_i]  \leq \bE{x\sim\p}{x_i^2}$ are at most $2$, and (ii)~call the mean testing algorithm (\cref{algo:testing:l2:unrolled}) on $\q$ with parameter $\dst/(2\sqrt{3\dims})$. The first step can be done coordinate-wise (checking each $\bE{x\sim\p}{x_i^2}$ by taking the empirical median, and concluding overall by a union bound over the $\dims$ estimates)\footnote{In more detail, for an arbitrary univariate Gaussian $X\sim\gaussian{\mu}{\sigma}$, testing $\bE{}{X^2} \leq 1$ vs. $\bE{}{X^2} > 2$ only takes, by considering the median, a constant number of samples (independent of $\mu,\sigma$).}{} with $O(\log\dims)$ samples in total; the second will cost $O(\frac{\sqrt{\dims}}{\dst^2})$ samples and be correct with probability at least $2/3$ by the foregoing discussion, so overall the test is correct with probability at  least $7/12$. Repeating independently a constant number of times and taking a majority vote then bring the success probability to the desired $2/3$.
\end{proof}

\paragraph{Acknowledgments.} The authors would like to thank Rajesh Jayaram, whose suggestions (in particular,~\cref{rk:improved:runtime}, and a tighter union bound for the completeness case of~\cref{theo:testing:l2:unrolled}) helped improve an earlier version of the paper.

\bibliographystyle{alpha}
\bibliography{biblio}

\newcommand{\etalchar}[1]{$^{#1}$}
\begin{thebibliography}{CKM{\etalchar{+}}19}

\bibitem[ABDK18]{AcharyaBDK18}
Jayadev Acharya, Arnab Bhattacharyya, Constantinos Daskalakis, and Saravanan
  Kandasamy.
\newblock Learning and testing causal models with interventions.
\newblock In {\em Advances in Neural Information Processing Systems 31},
  NeurIPS '18. Curran Associates, Inc., 2018.

\bibitem[ACK15a]{AcharyaCK15a}
Jayadev Acharya, Cl{\'e}ment~L. Canonne, and Gautam Kamath.
\newblock Adaptive estimation in weighted group testing.
\newblock In {\em Proceedings of the 2015 IEEE International Symposium on
  Information Theory}, ISIT '15, pages 2116--2120, Washington, DC, USA, 2015.
  IEEE Computer Society.

\bibitem[ACK15b]{AcharyaCK15b}
Jayadev Acharya, Cl{\'e}ment~L. Canonne, and Gautam Kamath.
\newblock A chasm between identity and equivalence testing with conditional
  queries.
\newblock In {\em Approximation, Randomization, and Combinatorial Optimization.
  Algorithms and Techniques.}, RANDOM '15, pages 449--466, Dagstuhl, Germany,
  2015. Schloss Dagstuhl--Leibniz-Zentrum fuer Informatik.

\bibitem[ADJ{\etalchar{+}}11]{AcharyaDJOP11}
Jayadev Acharya, Hirakendu Das, Ashkan Jafarpour, Alon Orlitsky, and Shengjun
  Pan.
\newblock Competitive closeness testing.
\newblock In {\em Proceedings of the 24th Annual Conference on Learning
  Theory}, COLT '11, pages 47--68, 2011.

\bibitem[ADK15]{AcharyaDK15}
Jayadev Acharya, Constantinos Daskalakis, and Gautam Kamath.
\newblock Optimal testing for properties of distributions.
\newblock In {\em Advances in Neural Information Processing Systems 28}, NIPS
  '15, pages 3577--3598. Curran Associates, Inc., 2015.

\bibitem[BBC{\etalchar{+}}19]{BezakovaBCSV19}
Ivona Bezakova, Antonio Blanca, Zongchen Chen, Daniel {\v{S}}tefankovi{\v{c}},
  and Eric Vigoda.
\newblock Lower bounds for testing graphical models: Colorings and
  antiferromagnetic {I}sing models.
\newblock In {\em Proceedings of the 32nd Annual Conference on Learning
  Theory}, COLT '19, pages 283--298, 2019.

\bibitem[BC17]{BatuC17}
Tu\u{g}kan Batu and Cl{\'e}ment~L. Canonne.
\newblock Generalized uniformity testing.
\newblock In {\em Proceedings of the 58th Annual IEEE Symposium on Foundations
  of Computer Science}, FOCS '17, pages 880--889, Washington, DC, USA, 2017.
  IEEE Computer Society.

\bibitem[BC18]{BhattacharyyaC18}
Rishiraj Bhattacharyya and Sourav Chakraborty.
\newblock Property testing of joint distributions using conditional samples.
\newblock {\em Transactions on Computation Theory}, 10(4):16:1--16:20, 2018.

\bibitem[BCG17]{BlaisCG17}
Eric Blais, Cl{\'e}ment~L. Canonne, and Tom Gur.
\newblock Distribution testing lower bounds via reductions from communication
  complexity.
\newblock In {\em Proceedings of the 32nd Computational Complexity Conference},
  CCC '17, pages 28:1--28:40, Dagstuhl, Germany, 2017. Schloss
  Dagstuhl--Leibniz-Zentrum fuer Informatik.

\bibitem[BDKR05]{BatuDKR05}
Tu\u{g}kan Batu, Sanjoy Dasgupta, Ravi Kumar, and Ronitt Rubinfeld.
\newblock The complexity of approximating the entropy.
\newblock {\em SIAM Journal on Computing}, 35(1):132--150, 2005.

\bibitem[BFF{\etalchar{+}}01]{BatuFFKRW01}
Tu\u{g}kan Batu, Eldar Fischer, Lance Fortnow, Ravi Kumar, Ronitt Rubinfeld,
  and Patrick White.
\newblock Testing random variables for independence and identity.
\newblock In {\em Proceedings of the 42nd Annual IEEE Symposium on Foundations
  of Computer Science}, FOCS '01, pages 442--451, Washington, DC, USA, 2001.
  IEEE Computer Society.

\bibitem[BFR{\etalchar{+}}00]{BatuFRSW00}
Tu\u{g}kan Batu, Lance Fortnow, Ronitt Rubinfeld, Warren~D. Smith, and Patrick
  White.
\newblock Testing that distributions are close.
\newblock In {\em Proceedings of the 41st Annual IEEE Symposium on Foundations
  of Computer Science}, FOCS '00, pages 259--269, Washington, DC, USA, 2000.
  IEEE Computer Society.

\bibitem[BFRV11]{BhattacharyyaFRV11}
Arnab Bhattacharyya, Eldar Fischer, Ronitt Rubinfeld, and Paul Valiant.
\newblock Testing monotonicity of distributions over general partial orders.
\newblock In {\em Proceedings of the 2nd Conference on Innovations in Computer
  Science}, ICS '11, pages 239--252, Beijing, China, 2011. Tsinghua University
  Press.

\bibitem[BKR04]{BatuKR04}
Tu\u{g}kan Batu, Ravi Kumar, and Ronitt Rubinfeld.
\newblock Sublinear algorithms for testing monotone and unimodal distributions.
\newblock In {\em Proceedings of the 36th Annual ACM Symposium on the Theory of
  Computing}, STOC '04, New York, NY, USA, 2004. ACM.

\bibitem[BV15]{BhattacharyaV15}
Bhaswar Bhattacharya and Gregory Valiant.
\newblock Testing closeness with unequal sized samples.
\newblock In {\em Advances in Neural Information Processing Systems 28}, NIPS
  '15, pages 2611--2619. Curran Associates, Inc., 2015.

\bibitem[BW18]{BalakrishnanW18}
Sivaraman Balakrishnan and Larry Wasserman.
\newblock Hypothesis testing for high-dimensional multinomials: A selective
  review.
\newblock {\em The Annals of Applied Statistics}, 12(2):727--749, 2018.

\bibitem[Can15a]{Canonne15b}
Cl{\'e}ment~L. Canonne.
\newblock Big data on the rise? - testing monotonicity of distributions.
\newblock In {\em Proceedings of the 42nd International Colloquium on Automata,
  Languages, and Programming}, ICALP '15, pages 294--305, 2015.

\bibitem[Can15b]{Canonne15a}
Cl{\'e}ment~L. Canonne.
\newblock A survey on distribution testing: Your data is big. but is it blue?
\newblock {\em Electronic Colloquium on Computational Complexity (ECCC)},
  22(63), 2015.

\bibitem[Can16]{Canonne16}
Cl{\'e}ment~L. Canonne.
\newblock Are few bins enough: Testing histogram distributions.
\newblock In {\em Proceedings of the 35th ACM SIGMOD-SIGACT-SIGART Symposium on
  Principles of Database Systems}, PODS '16, pages 455--463, New York, NY, USA,
  2016. ACM.

\bibitem[CDKS17]{CanonneDKS17}
Cl{\'e}ment~L. Canonne, Ilias Diakonikolas, Daniel~M. Kane, and Alistair
  Stewart.
\newblock Testing {B}ayesian networks.
\newblock In {\em Proceedings of the 30th Annual Conference on Learning
  Theory}, COLT '17, pages 370--448, 2017.

\bibitem[CDVV14]{ChanDVV14}
Siu~On Chan, Ilias Diakonikolas, Gregory Valiant, and Paul Valiant.
\newblock Optimal algorithms for testing closeness of discrete distributions.
\newblock In {\em Proceedings of the 25th Annual ACM-SIAM Symposium on Discrete
  Algorithms}, SODA '14, pages 1193--1203, Philadelphia, PA, USA, 2014. SIAM.

\bibitem[CFGM13]{ChakrabortyFGM13}
Sourav Chakraborty, Eldar Fischer, Yonatan Goldhirsh, and Arie Matsliah.
\newblock On the power of conditional samples in distribution testing.
\newblock In {\em Proceedings of the 4th Conference on Innovations in
  Theoretical Computer Science}, ITCS '13, pages 561--580, New York, NY, USA,
  2013. ACM.

\bibitem[CFGM16]{ChakrabortyFGM16}
Sourav Chakraborty, Eldar Fischer, Yonatan Goldhirsh, and Arie Matsliah.
\newblock On the power of conditional samples in distribution testing.
\newblock {\em SIAM Journal on Computing}, 45(4):1261--1296, 2016.

\bibitem[CJLW20]{chen2020learning}
Xi~Chen, Rajesh Jayaram, Amit Levi, and Erik Waingarten.
\newblock Learning and testing junta distributions with subcube conditioning.
\newblock {\em arXiv preprint arXiv:2004.12496}, 2020.

\bibitem[CKM{\etalchar{+}}19]{CanonneKMUZ19}
Cl\'ement~L. Canonne, Gautam Kamath, Audra McMillan, Jonathan Ullman, and Lydia
  Zakynthinou.
\newblock Private identity testing for high-dimensional distributions.
\newblock {\em arXiv preprint arXiv:1905.11947}, 2019.

\bibitem[CR14]{CanonneR14}
Cl{\'e}ment~L. Canonne and Ronitt Rubinfeld.
\newblock Testing probability distributions underlying aggregated data.
\newblock In {\em Proceedings of the 41st International Colloquium on Automata,
  Languages, and Programming}, ICALP '14, pages 283--295, 2014.

\bibitem[CRS14]{CanonneRS14}
Cl{\'e}ment~L. Canonne, Dana Ron, and Rocco~A. Servedio.
\newblock Testing equivalence between distributions using conditional samples.
\newblock In {\em Proceedings of the 25th Annual ACM-SIAM Symposium on Discrete
  Algorithms}, SODA '14, pages 1174--1192, Philadelphia, PA, USA, 2014. SIAM.

\bibitem[CRS15]{CanonneRS15}
Cl{\'e}ment~L. Canonne, Dana Ron, and Rocco~A. Servedio.
\newblock Testing probability distributions using conditional samples.
\newblock {\em SIAM Journal on Computing}, 44(3):540--616, 2015.

\bibitem[DDK18]{DaskalakisDK18}
Constantinos Daskalakis, Nishanth Dikkala, and Gautam Kamath.
\newblock Testing {I}sing models.
\newblock In {\em Proceedings of the 29th Annual ACM-SIAM Symposium on Discrete
  Algorithms}, SODA '18, pages 1989--2007, Philadelphia, PA, USA, 2018. SIAM.

\bibitem[DDS{\etalchar{+}}13]{DaskalakisDSVV13}
Constantinos Daskalakis, Ilias Diakonikolas, Rocco~A. Servedio, Gregory
  Valiant, and Paul Valiant.
\newblock Testing k-modal distributions: Optimal algorithms via reductions.
\newblock In {\em Proceedings of the 24th Annual ACM-SIAM Symposium on Discrete
  Algorithms}, SODA '13, pages 1833--1852, Philadelphia, PA, USA, 2013. SIAM.

\bibitem[DGPP18]{DiakonikolasGPP18}
Ilias Diakonikolas, Themis Gouleakis, John Peebles, and Eric Price.
\newblock Sample-optimal identity testing with high probability.
\newblock In {\em Proceedings of the 45th International Colloquium on Automata,
  Languages, and Programming}, ICALP '18, pages 41:1--41:14, 2018.

\bibitem[DK16]{DiakonikolasK16}
Ilias Diakonikolas and Daniel~M. Kane.
\newblock A new approach for testing properties of discrete distributions.
\newblock In {\em Proceedings of the 57th Annual IEEE Symposium on Foundations
  of Computer Science}, FOCS '16, pages 685--694, Washington, DC, USA, 2016.
  IEEE Computer Society.

\bibitem[DKN15]{DiakonikolasKN15a}
Ilias Diakonikolas, Daniel~M. Kane, and Vladimir Nikishkin.
\newblock Testing identity of structured distributions.
\newblock In {\em Proceedings of the 26th Annual ACM-SIAM Symposium on Discrete
  Algorithms}, SODA '15, pages 1841--1854, Philadelphia, PA, USA, 2015. SIAM.

\bibitem[DKP19]{DiakonikolasKP19}
Ilias Diakonikolas, Daniel~M. Kane, and John Peebles.
\newblock Testing identity of multidimensional histograms.
\newblock In {\em Proceedings of the 32nd Annual Conference on Learning
  Theory}, COLT '19, pages 1107--1131, 2019.

\bibitem[DKW18]{DaskalakisKW18}
Constantinos Daskalakis, Gautam Kamath, and John Wright.
\newblock Which distribution distances are sublinearly testable?
\newblock In {\em Proceedings of the 29th Annual ACM-SIAM Symposium on Discrete
  Algorithms}, SODA '18, pages 2747--2764, Philadelphia, PA, USA, 2018. SIAM.

\bibitem[DP17]{DaskalakisP17}
Constantinos Daskalakis and Qinxuan Pan.
\newblock Square {H}ellinger subadditivity for {B}ayesian networks and its
  applications to identity testing.
\newblock In {\em Proceedings of the 30th Annual Conference on Learning
  Theory}, COLT '17, pages 697--703, 2017.

\bibitem[FJO{\etalchar{+}}15]{FalahatgarJOPS15}
Moein Falahatgar, Ashkan Jafarpour, Alon Orlitsky, Venkatadheeraj Pichapati,
  and Ananda~Theertha Suresh.
\newblock Faster algorithms for testing under conditional sampling.
\newblock In {\em Proceedings of the 28th Annual Conference on Learning
  Theory}, COLT '15, pages 607--636, 2015.

\bibitem[FLV17]{FischerLV17}
Eldar Fischer, Oded Lachish, and Yadu Vasudev.
\newblock Improving and extending the testing of distributions for
  shape-restricted properties.
\newblock In {\em Proceedings of the 34th Symposium on Theoretical Aspects of
  Computer Science}, STACS '17, pages 31:1--31:14, Dagstuhl, Germany, 2017.
  Schloss Dagstuhl--Leibniz-Zentrum fuer Informatik.

\bibitem[GGR96]{GoldreichGR96}
Oded Goldreich, Shafi Goldwasser, and Dana Ron.
\newblock Property testing and its connection to learning and approximation.
\newblock In {\em Proceedings of the 37th Annual IEEE Symposium on Foundations
  of Computer Science}, FOCS '96, pages 339--348, Washington, DC, USA, 1996.
  IEEE Computer Society.

\bibitem[GLP18]{GheissariLP18}
Reza Gheissari, Eyal Lubetzky, and Yuval Peres.
\newblock Concentration inequalities for polynomials of contracting {I}sing
  models.
\newblock {\em Electronic Communications in Probability}, 23(76):1--12, 2018.

\bibitem[GMV06]{GuhaMV06}
Sudipto Guha, Andrew McGregor, and Suresh Venkatasubramanian.
\newblock Streaming and sublinear approximation of entropy and information
  distances.
\newblock In {\em Proceedings of the 17th Annual ACM-SIAM Symposium on Discrete
  Algorithms}, SODA '06, pages 733--742, Philadelphia, PA, USA, 2006. SIAM.

\bibitem[Gol17]{Goldreich17}
Oded Goldreich.
\newblock {\em Introduction to Property Testing}.
\newblock Cambridge University Press, 2017.

\bibitem[GR00]{GoldreichR00}
Oded Goldreich and Dana Ron.
\newblock On testing expansion in bounded-degree graphs.
\newblock {\em Electronic Colloquium on Computational Complexity (ECCC)},
  7(20), 2000.

\bibitem[GTZ17]{GouleakisTZ17}
Themistoklis Gouleakis, Christos Tzamos, and Manolis Zampetakis.
\newblock Faster sublinear algorithms using conditional sampling.
\newblock In {\em Proceedings of the 28th Annual ACM-SIAM Symposium on Discrete
  Algorithms}, SODA '17, pages 1743--1757, Philadelphia, PA, USA, 2017. SIAM.

\bibitem[GTZ18]{GouleakisTZ18}
Themis Gouleakis, Christos Tzamos, and Manolis Zampetakis.
\newblock Certified computation from unreliable datasets.
\newblock In {\em Proceedings of the 31st Annual Conference on Learning
  Theory}, COLT '18, pages 3271--3294, 2018.

\bibitem[ILR12]{IndykLR12}
Piotr Indyk, Reut Levi, and Ronitt Rubinfeld.
\newblock Approximating and testing k-histogram distributions in sub-linear
  time.
\newblock In {\em Proceedings of the 31st ACM SIGMOD-SIGACT-SIGART Symposium on
  Principles of Database Systems}, PODS '12, pages 15--22, New York, NY, USA,
  2012. ACM.

\bibitem[Kam18]{Kamath18}
Gautam Kamath.
\newblock {\em Modern Challenges in Distribution Testing}.
\newblock PhD thesis, Massachusetts Institute of Technology, September 2018.

\bibitem[KMS18]{KhotMS18}
Subhash Khot, Dor Minzer, and Muli Safra.
\newblock On monotonicity testing and {B}oolean isoperimetric-type theorems.
\newblock {\em SIAM Journal on Computing}, 47(6):2238--2276, 2018.

\bibitem[KT19]{KamathT19}
Gautam Kamath and Christos Tzamos.
\newblock Anaconda: A non-adaptive conditional sampling algorithm for
  distribution testing.
\newblock In {\em Proceedings of the 30th Annual ACM-SIAM Symposium on Discrete
  Algorithms}, SODA '19, Philadelphia, PA, USA, 2019. SIAM.

\bibitem[NS02]{NaorS02}
Assaf Naor and Gideon Schechtman.
\newblock Remarks on non-linear type and {P}isier's inequality.
\newblock {\em Journal fur die Reine und Angewandte Mathematik}, 552:213--236,
  2002.

\bibitem[OS18]{OnakS18}
Krzysztof Onak and Xiaorui Sun.
\newblock Probability-revealing samples.
\newblock In {\em Proceedings of the 21st International Conference on
  Artificial Intelligence and Statistics}, AISTATS '18, pages 84:2018--2026.
  JMLR, Inc., 2018.

\bibitem[Pan08]{Paninski08}
Liam Paninski.
\newblock A coincidence-based test for uniformity given very sparsely sampled
  discrete data.
\newblock {\em IEEE Transactions on Information Theory}, 54(10):4750--4755,
  2008.

\bibitem[Pis86]{Pisier86}
Gilles Pisier.
\newblock Probabilistic methods in the geometry of {B}anach spaces.
\newblock In Giorgio Letta and Maurizio Pratelli, editors, {\em Probability and
  Analysis}, pages 167--241. Springer, 1986.

\bibitem[RS09]{RubinfeldS09}
Ronitt Rubinfeld and Rocco~A. Servedio.
\newblock Testing monotone high-dimensional distributions.
\newblock {\em Random Structures and Algorithms}, 34(1):24--44, 2009.

\bibitem[Rub12]{Rubinfeld12}
Ronitt Rubinfeld.
\newblock Taming big probability distributions.
\newblock {\em XRDS}, 19(1):24--28, 2012.

\bibitem[SSJ17]{SardharwallaSJ17}
Imdad S.~B. Sardharwalla, Sergii Strelchuk, and Richard Jozsa.
\newblock Quantum conditional query complexity.
\newblock {\em Quantum Information \& Computation}, 17(7\& 8):541--566, 2017.

\bibitem[Tal93]{Talagrand93}
Michel Talagrand.
\newblock Isoperimetry, logarithmic {S}obolev inequalities on the discrete
  cube, and {M}argulis' graph connectivity theorem.
\newblock {\em Geometric \& Functional Analysis}, 3(3):295--314, 1993.

\bibitem[Val11]{Valiant11}
Paul Valiant.
\newblock Testing symmetric properties of distributions.
\newblock {\em SIAM Journal on Computing}, 40(6):1927--1968, 2011.

\bibitem[VV14]{ValiantV14}
Gregory Valiant and Paul Valiant.
\newblock An automatic inequality prover and instance optimal identity testing.
\newblock In {\em Proceedings of the 55th Annual IEEE Symposium on Foundations
  of Computer Science}, FOCS '14, pages 51--60, Washington, DC, USA, 2014. IEEE
  Computer Society.

\bibitem[VV17]{ValiantV17a}
Gregory Valiant and Paul Valiant.
\newblock An automatic inequality prover and instance optimal identity testing.
\newblock {\em SIAM Journal on Computing}, 46(1):429--455, 2017.

\bibitem[Wag15]{Waggoner15}
Bo~Waggoner.
\newblock $l_p$ testing and learning of discrete distributions.
\newblock In {\em Proceedings of the 6th Conference on Innovations in
  Theoretical Computer Science}, ITCS '15, pages 347--356, New York, NY, USA,
  2015. ACM.

\end{thebibliography}

 \appendix
 \section{Proof of~\cref{lemma:variance:z:general:p}}
\label{sec:mean-appendix}
  We start by computing the second moment of the statistic:
\begin{align}
  \bE{}{(\bZ(\p))^2} &= \bE{}{\dotprod{\bar{\bX}}{ \bar{\bY}}\dotprod{\bar{\bX}}{ \bar{\bY}}}
  = \sum_{1\leq i,j\leq \dimss} \bE{}{ \bar{\bX}_i\bar{\bX}_j }\bE{}{ \bar{\bY}_i\bar{\bY}_j }
  = \sum_{1\leq i,j\leq \dimss} \bE{}{ \bar{\bX}_i\bar{\bX}_j }^2 \notag\\
  &= \sum_{i=1}^\dimss \bE{}{ \bar{\bX}_i^2 }^2 + 2\sum_{i<j} \bE{}{ \bar{\bX}_i\bar{\bX}_j }^2
  = \frac{\dimss}{\ns^2} + 2\frac{\ns-1}{\ns^2}\normtwo{\mu(\p)}^2+\frac{(\ns-1)^2}{\ns^2}\norm{\mu(\p)}_4^4 + 2\sum_{i<j} \bE{}{ \bar{\bX}_i\bar{\bX}_j }^2\notag\\
  &\leq \frac{\dimss}{\ns^2} + \frac{2}{\ns}\normtwo{\mu(\p)}^2+\norm{\mu(\p)}_4^4 + 2\sum_{i<j} \bE{}{ \bar{\bX}_i\bar{\bX}_j }^2. \label{eq:variance:z}
\end{align}
The final equality follows since, for all $i\in[\dimss]$,
\begin{equation}\label{eq:expectation:barxisquared}
\ns^2\bE{}{ \bar{\bX}_i^2 } = \sum_{k=1}^\ns \bE{}{ {\bX^{(k)}_i}^2 }+ 2\sum_{1\leq k<\ell\leq \ns}\bE{}{ \bX^{(k)}_i \bX^{(\ell)}_i }
  = \ns + \ns(\ns-1)\mu(\p)_i^2
\end{equation}

  For any $i<j$, let $\sigma(\p)_{ij} \eqdef \bE{x\sim\p}{x_ix_j}$. We have
\begin{align*}
  \bE{}{ \bar{\bX}_i\bar{\bX}_j }
  &= \frac{1}{\ns^2}\sum_{1\leq k,\ell\leq \ns} \bE{}{ \bX^{(k)}_i \bX^{(\ell)}_j }
  = \frac{1}{\ns^2}\Paren{ \sum_{k=1}^\ns \bE{}{ \bX^{(k)}_i \bX^{(k)}_j } +  2\sum_{1\leq k<\ell\leq \ns} \bE{}{ \bX^{(k)}_i } \bE{}{ \bX^{(\ell)}_j } } \\
  &= \frac{1}{\ns}\sigma(\p)_{ij} + \frac{\ns-1}{\ns} \mu(\p)_i\mu(\p)_j\,,
\end{align*}
so that
\begin{align*}
2\sum_{i<j} \bE{}{ \bar{\bX}_i\bar{\bX}_j }^2
  &= \frac{1}{\ns^2}\sum_{i\neq j}\sigma(\p)_{ij}^2 + \frac{(\ns-1)^2}{\ns^2}\sum_{i\neq j}\mu(\p)_i^2\mu(\p)_j^2 + \frac{2(\ns-1)}{\ns^2}\sum_{i\neq j}\sigma(\p)_{ij}\mu(\p)_i\mu(\p)_j\\
  &\leq \frac{1}{\ns^2}\sum_{i\neq j}\sigma(\p)_{ij}^2 + \normtwo{\mu(\p)}^4-\norm{\mu(\p)}_4^4 + \frac{2}{\ns}\sum_{i\neq j}\sigma(\p)_{ij}\mu(\p)_i\mu(\p)_j\,.
\end{align*}

Combining this with~\eqref{eq:expectation:z} and~\eqref{eq:variance:z}, we have 
\begin{align}\label{eq:variance:z:partial}
  \Var[\bZ(\p)] &= \bE{}{(\bZ(\p))^2} - \bE{}{\bZ(\p)}^2 
  \leq \frac{\dimss}{\ns^2} + \frac{2}{\ns}\normtwo{\mu(\p)}^2 
  + \frac{1}{\ns^2}\sum_{i\neq j}\sigma(\p)_{ij}^2 
  + \frac{2}{\ns}\sum_{i\neq j}\sigma(\p)_{ij}\mu(\p)_i\mu(\p)_j 
\end{align}
We use the Cauchy-Schwarz inequality to simplify the last term:
\begin{align*}
  \Var [\bZ(\p)]
  &\leq \frac{\dimss}{\ns^2} + \frac{2}{\ns}\normtwo{\mu(\p)}^2 
  + \frac{1}{\ns^2}\sum_{i\neq j}\sigma(\p)_{ij}^2 
  + \frac{2}{\ns}\sqrt{\sum_{i\neq j}\sigma(\p)_{ij}^2}\sqrt{\sum_{i\neq j}\mu(\p)_i^2\mu(\p)_j^2}  \\
  &= \frac{\dimss}{\ns^2} + \frac{2}{\ns}\normtwo{\mu(\p)}^2 
  + \frac{1}{\ns^2}\sum_{i\neq j}\sigma(\p)_{ij}^2 
  + \frac{2}{\ns}\sqrt{\sum_{i\neq j}\sigma(\p)_{ij}^2}\sqrt{\normtwo{\mu(\p)}^4-\norm{\mu(\p)}_4^4}  \\
  &\leq \frac{\dimss}{\ns^2} + \frac{2}{\ns}\normtwo{\mu(\p)}^2 
  + \frac{1}{\ns^2}\sum_{i\neq j}\sigma(\p)_{ij}^2 
  + \frac{2}{\ns}\normtwo{\mu(\p)}^2\sqrt{\sum_{i\neq j}\sigma(\p)_{ij}^2}
\end{align*}
  Letting $\Sigma(\p) \eqdef \bE{x\sim\p}{xx^T}$, we have 
\[
  \Sigma(\p)_{ij} = \begin{cases}1 &\text{ if } i=j \\ \sigma(\p)_{ij} &\text {otherwise}\end{cases}
\]
and thus can rewrite the above as
\begin{align*}
  \Var [\bZ(\p)]
  &\leq \frac{1}{\ns^2}\norm{\Sigma(\p)}_F^2 + \frac{2}{\ns}\normtwo{\mu(\p)}^2  + \frac{2}{\ns}\normtwo{\mu(\p)}^2\sqrt{\norm{\Sigma(\p)}_F^2-\Tr[\Sigma(\p)]} \nonumber \\
  &\leq \frac{1}{\ns^2}\norm{\Sigma(\p)}_F^2 + \frac{2}{\ns}\normtwo{\mu(\p)}^2  + \frac{2}{\ns}\normtwo{\mu(\p)}^2\norm{\Sigma(\p)}_F \\
  &\leq \frac{1}{\ns^2}\norm{\Sigma(\p)}_F^2 + \frac{4}{\ns}\normtwo{\mu(\p)}^2\norm{\Sigma(\p)}_F,
\end{align*}
as desired.

\end{document}